\documentclass[11pt,letterpaper,english]{article}
\usepackage[utf8]{inputenc}			% Allows the use of UTF-8 characters
\usepackage{makeidx}
\usepackage[pdftex]{graphicx}
\usepackage{mathtools,amsfonts}
\usepackage{array}
\usepackage{multirow, hhline}

\usepackage{nicefrac}
\usepackage{tikz}
\usetikzlibrary{arrows}
\usetikzlibrary{calc}
\usepackage{paralist}

\usepackage{wrapfig}
\usepackage[english]{babel}		% Allows hiphenization of words and other stuff
\usepackage[margin=1in]{geometry}	% Defines margins to 1 inch
\usepackage{lipsum}			% provides filler text \lipsum[n]
\usepackage{comment}
\usepackage{chngpage}			% Allows width to be changed dynamically
\usepackage{amsfonts,amssymb,amsmath, amsthm}
\usepackage{amsopn}
\usepackage{color}
\usepackage{pgfplots}
\usepackage{pgfgantt}
\usepackage{algpseudocode}
\usepackage[Algorithm]{algorithm}
\usepackage{subcaption}
\usepackage{tabto}
\usepackage{csquotes} %quotations
\usepackage{multirow}
\usepackage{amsthm}
\usepackage[colorlinks,urlcolor=black,citecolor=black,linkcolor=black,menucolor=black]{hyperref}
\newtheorem{theorem}{Theorem}

\newtheorem{observation}[theorem]{Observation}
\newtheorem{definition}[theorem]{Definition}

\newtheorem{lemma}[theorem]{Lemma}

\newtheorem*{problem*}{Problem} 

\newtheorem*{fact*}{Fact}

\newcounter{example}[section]

\title{Improving EFX Guarantees through Rainbow Cycle Number}
\author{Bhaskar Ray Chaudhury\thanks{MPI for Informatics, Saarland Informatics Campus, Graduate School of Computer Science, Saarbr\"ucken, Germany}\\ \texttt{\small braycha@mpi-inf.mpg.de} \and Jugal Garg\thanks{University of Illinois at Urbana-Champaign. Supported by NSF Grant CCF-1942321 (CAREER)}\\ \texttt{\small jugal@illinois.edu}  \and Kurt Mehlhorn \thanks{MPI for Informatics, Saarland Informatics Campus, Germany}\\ \texttt{\small mehlhorn@mpi-inf.mpg.de} \and Ruta Mehta\thanks{University of Illinois at Urbana-Champaign. Supported by NSF Grant CCF-1750436 (CAREER)}\\ \texttt{\small rutameht@illinois.edu} \and Pranabendu Misra \thanks{MPI for Informatics, Saarland Informatics Campus, Germany}\\ \texttt{\small pmisra@mpi-inf.mpg.de}}

\usepackage{booktabs} % For formal tables
\newcommand{\abs}[1]{| #1 |}
\newcommand{\set}[2]{ \{ #1 \mid  #2 \} }
\newcommand{\sset}[1]{\left\{ #1 \right\} }
\newcommand{\eps}{\varepsilon}
\newcommand{\ceil}[1]{\lceil #1 \rceil}

\newcommand{\ram}{{\mathit{R}}}
\newcommand{\pd}{>_{\mathit{PD}}}

\renewcommand{\O}{\mathcal{O}}

\begin{document}

\maketitle 

\begin{abstract}
	We study the problem of fairly allocating a set of indivisible goods among $n$ agents with additive valuations. Envy-freeness up to \emph{any} good (EFX) is arguably the most compelling fairness notion in this context. 
	However, the existence of EFX allocations has not been settled and is one of the most important problems in fair division~\cite{ProcacciaCACM20}. Towards resolving this problem, many impressive results show the existence of its relaxations. In particular,~\cite{AmanatidisMN20} shows the existence of $0.618$-EFX allocations, and~\cite{CKMS20} shows that EFX allocation exists if we do not allocate at most $n-1$ goods. The latter result was recently improved for three agents in~\cite{CGM20}, in which the two unallocated goods are allocated through an involved procedure. Reducing the number of unallocated goods for an arbitrary number of agents is a systematic way to settle the big question. 
	
	In this paper, we develop a new approach, and show that for every $\eps \in (0,1/2]$, there always exists a $(1-\eps)$-EFX allocation with \emph{sublinear} number of unallocated goods and high Nash welfare. 
	For this, we reduce the EFX problem to a novel problem in extremal graph theory. We introduce the notion of \emph{rainbow cycle number} $\ram(\cdot)$. For all $d \in \mathbb{N}$, $\ram(d)$ is the largest $k$ such that there exists a $k$-partite graph $G =(\cup_{i \in [k]} V_i, E)$, in which 
	\begin{itemize}
		\item each part has at most $d$ vertices, i.e., $\abs{V_i} \leq d$ for all $i \in [k]$,
		\item for any two parts $V_i$ and $V_j$, each vertex in $V_i$ has an incoming edge from some vertex in $V_j$ and vice-versa, and
		\item there exists no cycle in $G$ that contains at most one vertex from each part.
	\end{itemize}
	We show that any upper bound on $\ram(d)$ directly translates to a sublinear bound on the number of unallocated goods. We establish a polynomial upper bound on $\ram(d)$, yielding our main result. Furthermore, our approach is constructive, which also gives a polynomial-time algorithm for finding such an allocation.
\end{abstract}

\clearpage 

\section{Introduction}\label{intro}
Fair division of resources is a fundamental problem in many disciplines, including computer science, economics, and social choice theory. The objective is to distribute resources among agents in a \emph{fair} (no agent is significantly unhappy with her allocation) and \emph{efficient} (there is no other \emph{fair allocation} that can achieve better total welfare) manner. Mentions of such problems date back to the Bible and ancient Greek mythology. Today the issue of fair division arises in division of labor, inheritance, or computing resources, divorce settlements, partnership dissolutions, splitting rent among tenants, splitting taxi fare among passengers, dividing household tasks, air traffic management, frequency allocation, and so on. 
In the internet age, the existence of several centralized platforms and more computational power has triggered substantial interest from the economics and computer science community to find computationally tractable protocols to allocate resources fairly; see Spliddit~\cite{spliddit} and Fair Outcomes~\cite{fairoutcome} for more details on fair division protocols used in real-life scenarios.

\paragraph{Discrete Fair Division.} In this paper, we focus on one of the most important open problems in discrete fair division. To this end, we first describe a typical setup of a problem instance: Given a set $N$ of $n$ agents and a set $M$ of $m$ \emph{indivisible} goods, the goal is to determine a partition $X = \langle X_1,X_2, \dots ,X_n \rangle$ of the good set $M$ such that agent $i \in N$ receives the bundle $X_i$ and the allocation is \emph{fair}. Depending on the notion of fairness, there are plethora of problems in this setting. 

\paragraph{EFX Allocations.} A quintessential notion of fairness is that of envy-freeness (EF): an allocation $X$ is said to be envy-free if and only if for every pair of agents $i$ and $j$ we have $v_i(X_i) \geq v_i(X_j)$, i.e., each agent $i$, values her own bundle at least as much as she values the bundles of other agents. However, such allocations may not always exist: consider a simple example with two agents having a positive valuation towards a single good. The agent that gets this good is envied by the one that does not. Therefore, several relaxations of envy-freeness have been proposed and studied over the last fifteen years~\cite{LiptonMMS04,budish2011combinatorial,CaragiannisKMP016}. The most compelling relaxation is \emph{envy-freeness up to any good (EFX)}, where no agent envies the other agent following the removal of \emph{any} single good from the other's bundle; that is, an allocation $X$ is said to be EFX if and only if for every pair of agents $i$ and $j$ we have $v_i(X_i) \geq v_i(X_j \setminus \left\{g\right\})$ for all $g \in X_j$. It is also regarded as the best analogue of envy-freeness in discrete fair division: Caragiannis et al.~\cite{CaragiannisGravin19} remarked that,

\begin{quote}
	``\textit{Arguably, EFX is the best fairness analog of envy-freeness for indivisible items.}'' 
\end{quote}

Unfortunately, it is not known whether EFX allocations always exist, even when there are only four agents with additive valuations despite significant efforts by many researchers, e.g., see ~\cite{CaragiannisKMP016, Moulin19}. Indeed, only recently was this question resolved affirmatively for three agents with additive valuations~\cite{CGM20}\footnote{Recipient of the Exemplary Theory Paper Award and the Best Paper with a Student Lead Author Award at ACM EC 2020}. In fact, Procaccia ~\cite{ProcaciaCACM} remarked that,

\begin{quote}
	``\textit{This fundamental and deceptively accessible question is open. In my view, it is the successor of envy-free cake cutting as fair division's biggest problem.}'' 
\end{quote}

There has been a substantial study on the existence of EFX allocation in special cases and its relaxations. For instance, EFX allocations exist when agents' valuations are identical~\cite{TimPlaut18}, binary~\cite{darmann2014binary,barman2018binarynsw}, and bi-valued~\cite{amanatidis2020mnwefx}. The two primary relaxations of EFX are approximate-EFX allocations and partial-EFX allocations:
\begin{itemize}
	\item \textbf{Approximate-EFX Allocation:} An allocation $X = \langle X_1,X_2, \dots, X_n \rangle$ is an $\alpha$-EFX allocation %(an $\alpha$-approximate-EFX allocation) 
	for some scalar $\alpha \in (0,1]$, if for every pair of agents $i$ and $j$, we have $v_i(X_i) \geq \alpha \cdot v_i(X_j \setminus \left\{g\right\})$ for all $g \in X_j$. Plaut and Roughgarden~\cite{TimPlaut18} showed the existence of $\tfrac{1}{2}$-EFX allocations. A clever modification of the same approach leads to a $0.618$-EFX allocation~\cite{AmanatidisMN20}. %where $\phi$ is the golden ratio.   
	\item \textbf{Partial-EFX Allocation:} An allocation $X = \langle X_1,X_2, \dots , X_n \rangle$ is called a partial-EFX allocation if $X$ is EFX and not all goods are necessarily allocated, i.e., $\cup_{i \in [n]} X_i \subseteq M$. There is always a trivial partial EFX allocation where each $X_i$ is empty. Therefore, a good partial EFX allocation is the one which has good \emph{qualitative} and \emph{quantitative} guarantees on the unallocated goods. Caragiannis et al.~\cite{CaragiannisGravin19} showed that there exists a partial EFX allocation where every agent gets a bundle that she values at least as much as half of her value for the bundle she receives in a \emph{Nash welfare} maximizing allocation. Here, the Nash welfare of an allocation $\mathit{NW}(X) = (\prod_{i \in [n]} v_i(X_i) )^{\nicefrac{1}{n}}$ is another popular measure of fairness and economic efficiency. Following the same line of work, Chaudhury et al.~\cite{CKMS20} showed that there always exists a partial EFX allocation $X$ and a set of unallocated goods $P$ such that 
	\begin{itemize}
		\item \emph{nobody envies the set of unallocated items:} $v_i(X_i) \geq v_i(P)$ for all $i \in N$, and  
		\item \emph{at most $n-1$ goods are unallocated:} $\lvert P \rvert \le n-1$. 
	\end{itemize}  
\end{itemize}

We remark that studying relaxations (of EFX allocations) is a systematic and promising direction to investigate the existence of EFX allocations. It has been suspected in Plaut and Roughgarden~\cite{TimPlaut18} that EFX allocations may not exist in the general setting: 
\begin{quote}
	``\textit{We suspect that at least for general valuations, there exist instances where no EFX allocation exists, and it may be easier to find a counterexample in that setting.}'' 
\end{quote}
However, finding counter-examples, at least in the additive setting, seems to be a very challenging task; quite recently Manurangsi and Suksompong~\cite{ManurangsiS20} showed that when agents valuations for individual items are drawn at random from a probability distribution, then EFX allocations exist with high probability. This demands a non brute-force approach to find counter-examples, if any. Thus finding better relaxations (improving the approximation factor or reducing the number of unallocated goods in a partial EFX allocation) is a crucial step to find the right answer to this big open question. We achieve exactly this by our first main result,

\begin{theorem}
	\label{mainthm1}
	For all $\eps \in (0,1/2]$ we can determine a partial allocation $X$ and a set of unallocated goods $P$ in polynomial time such that
	\begin{itemize}
		\item $X$ is $(1- \varepsilon)$-EFX,
		\item $\abs{P} \le 64 (n / \eps)^{4 / 5}$.
	\end{itemize}
\end{theorem}   

We remark that reducing the number of unallocated goods could be quite challenging: Indeed, a corollary from the main result in~\cite{CKMS20} already establishes that there exists a partial EFX allocation with at most two goods unallocated when there are three agents and at most three goods when there are four agents. However, removing the last two goods to obtain an EFX allocation for three agents turns out to be highly non-trivial task and the proof in ~\cite{CGM20} requires careful and cumbersome case analysis. Quite recently, Berger et al.~\cite{BCFF'21} overcome similar challenges and show the existence of EFX allocations with at most $n-2$ unallocated goods, and the existence of EFX allocations with at most one unallocated good when there are four agents.\footnote{Their proof also works for a broader class of valuation functions called \emph{nice cancellable valuations}. A valuation $v$ is a nice cancellable valuation if (i) $v(A) \neq v(B)$ for all $A \neq B$ and (ii) for all $A, B \subseteq M$ and $g \in M \setminus (A \cup B)$, $v(A \cup \{g\}) > v(B \cup \{g\})$ implies that $v(A) > v(B) $.} Their proof is also very involved and spans over 18 pages.  Furthermore, in Section~\ref{limitationoftechnique} of this paper, we show that the techniques in~\cite{CGM20, BCFF'21} does not extend to four agents with additive valuations for finding a $(1-\varepsilon)$-EFX allocation. 

\emph{In this paper, we develop a novel method that reduces the problem of determining good relaxations of EFX allocations to a combinatorial problem in graph theory.} We call it the \emph{rainbow cycle number} of an integer, defined as follows. 

\begin{definition}
	\label{Rdef}
	For any positive integer $d$, the rainbow cycle number or  $\ram(d)$ is the largest $k$ such that there exists a directed $k$-partite graph $G = (\cup _{i \in [k]} V_i, E)$ such that 
	\begin{enumerate}
		\item $\abs{V_i}\leq d$ for all $i \in [k]$, 
		\item for any two distinct parts $V_i$ and $V_j$ in $G$, every vertex in $V_i$ has an incoming edge from a vertex in $V_j$, and 
		\item there exists no cycle in $G$ that intersects each part at most once.
	\end{enumerate}
\end{definition}

Let us deduce that $\ram(1)=1$: It is clear that $G$ can be a single vertex and satisfy all the conditions in Definition~\ref{Rdef} and thus $\ram(1) \geq 1$. However, $\ram(1)$ cannot be larger than one, as otherwise we have two parts $V_1$ and $V_2$ in a graph $G$, where there is exactly one vertex each in $V_1$ and $V_2$. So let $V_1 = \left\{a_1\right\}$ and $V_2 = \left\{a_2\right\}$. By condition 2 in Definition~\ref{Rdef}, we must have an edge from $a_1$ to $a_2$ and an edge from $a_2$ to $a_1$. This gives a 2-cycle $a_1 \rightarrow a_2 \rightarrow a_1$. However, this cycle contains exactly one vertex from each $V_1$ and $V_2$, which contradicts condition 3 in Definition~\ref{Rdef}.

Similarly, using a more involved argument (appearing below) we can also determine that $R(2) = 2$. However, it is not at all clear what values $R(d)$ takes, or if it is finite for all integers $d$. 
A key technical result of this paper is a polynomial (in $d$) upper-bound on $\ram(d)$.

\begin{theorem}
	\label{mainthm3}
	For all $d \ge 1$, we have $\ram(d) \leq d^4 + d$. Furthermore, let $G$ be a $k$-partite digraph with $k > d^4 + d$ parts of cardinality at most $d$ each, such that for every vertex $v$ and any part $W$ not containing $v$, there is an edge from $W$ to $v$. Then, there exists a cycle in $G$ visiting each part at most once, and it can be found in time polynomial in $k$. 
\end{theorem}

Observe that the definition of the rainbow cycle number ($\ram( \cdot )$) is independent of the agents, goods and valuation functions. In the second key result of this paper, we establish a direct relation between the rainbow cycle number and the existence of better EFX relaxations: Finding a good upper bound on the rainbow cycle number can get us weaker relaxations of EFX allocations (we can asymptotically improve the number of unallocated goods). Formally, 

\begin{theorem}
	\label{mainthm2}
	Let $h(d) = d \cdot \ram(d)$ and $\eps \in (0,1/2]$.  Let $h^{-1}(n/\eps)$ be the smallest integer such that $h(d) \ge n/\eps$. Then, there is a $(1- \varepsilon)$-EFX allocation $X$ and a set of unallocated goods $P$ such that $\abs{P} \leq ({4n} / {(\eps \cdot h^{-1}(2n / \varepsilon))}$. 
\end{theorem}

Theorems~\ref{mainthm3} and~\ref{mainthm2} imply Theorem~\ref{mainthm1}. We remark that, although we give a polynomial upper bound on $\ram(d)$, we believe that there is further room for improvement. As an illustration, we briefly show that $\ram(2) \leq 2$, which is significantly better than our upper-bound for $d=2$ obtained from Theorem~\ref{mainthm2}. %(our proposed bound being $16 +2 = 18$). 
We prove this by contradiction. Let us assume otherwise and let $V_1$, $V_2$ and $V_3$ be any three parts of $G$. We first look into the edges of the induced bipartite graph $G[V_1 \cup V_2]$. Without loss of generality, let us assume that vertex $b_1$ in $V_2$ has an incoming edge from vertex $a_1$ in $V_1$. By condition 2 in Definition~\ref{Rdef}, $a_1$ has an incoming edge from some vertex in $V_2$. However, this vertex cannot be $b_1$ as this will violate condition 3 in Definition~\ref{Rdef}. This implies that there must be another vertex in $V_2$, say $b_2$ that has an edge to $a_1$.  Again, by a similar argument, $b_2$ cannot have an incoming edge from $a_1$ and therefore has an incoming edge from another vertex in $V_1$, say $a_2$ and $a_2$ has the incoming edge from $b_1$ and not $b_2$ (since there can be no other vertices in $V_2$). Thus, the induced bipartite graph $G[V_1 \cup V_2]$ is a four-cycle as shown below 
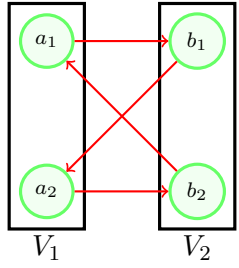
\begin{center}
	\begin{tikzpicture}
	[
	agent/.style={circle, draw=green!60, fill=green!5, very thick},
	good/.style={circle, draw=red!60, fill=red!5, very thick, minimum size=1pt},
	]
	%Parts
	\draw[black, very thick] (-0.5,0.5) rectangle (0.5,-2.5);
	\draw[black, very thick] (-0.5+2,0.5) rectangle (0.5+2,-2.5);
	\node at (0,-2.75) {$V_1$};
	\node at (2,-2.75) {$V_2$};
	
	%Vertices for config1
	\node[agent]      (a1) at (0,0)      {$\scriptstyle{a_1}$};
	\node[agent]      (a2) at (0,-2)      {$\scriptstyle{a_2}$};
	\node[agent]      (b1) at (2,0)     {$\scriptstyle{b_1}$};
	\node[agent]      (b2) at (2,-2)     {$\scriptstyle{b_2}$};

	%Edges
	\draw[->,red,thick] (a1) -- (b1);
	\draw[->,red,thick] (b2) -- (a1);
	\draw[->,red,thick] (a2) -- (b2);
	\draw[->,red,thick] (b1) -- (a2);
	\end{tikzpicture}	
\end{center}

Note that the induced bipartite graph $G[V_2 \cup V_3]$ will be isomorphic to $G[V_1 \cup V_2]$.  Thus, so far we have the following edges in $G[V_1 \cup V_2 \cup V_3]$,
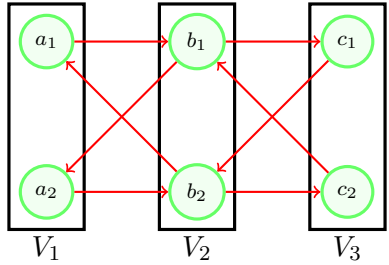
\begin{center}
	\begin{tikzpicture}
	[
	agent/.style={circle, draw=green!60, fill=green!5, very thick},
	good/.style={circle, draw=red!60, fill=red!5, very thick, minimum size=1pt},
	]
	%Parts
	\draw[black, very thick] (-0.5,0.5) rectangle (0.5,-2.5);
	\draw[black, very thick] (-0.5+2,0.5) rectangle (0.5+2,-2.5);
	\draw[black, very thick] (-0.5+4,0.5) rectangle (0.5+4,-2.5);
	\node at (0,-2.75) {$V_1$};
	\node at (2,-2.75) {$V_2$};
	\node at (4,-2.75) {$V_3$};	
	
	%Vertices for config1
	\node[agent]      (a1) at (0,0)      {$\scriptstyle{a_1}$};
	\node[agent]      (a2) at (0,-2)      {$\scriptstyle{a_2}$};
	\node[agent]      (b1) at (2,0)     {$\scriptstyle{b_1}$};
	\node[agent]      (b2) at (2,-2)     {$\scriptstyle{b_2}$};
	\node[agent]      (c1) at (4,0)     {$\scriptstyle{c_1}$};
	\node[agent]      (c2) at (4,-2)     {$\scriptstyle{c_2}$};
	
	%Edges
	\draw[->,red,thick] (a1) -- (b1);
	\draw[->,red,thick] (b2) -- (a1);
	\draw[->,red,thick] (a2) -- (b2);
	\draw[->,red,thick] (b1) -- (a2);
	
	\draw[->,red,thick] (b1) -- (c1);
	\draw[->,red,thick] (c2) -- (b1);
	\draw[->,red,thick] (b2) -- (c2);
	\draw[->,red,thick] (c1) -- (b2);
	\end{tikzpicture}	
\end{center}
We now look at the edges between the parts $V_1$ and $V_3$. Since $G[V_1 \cup V_3]$ is isomorphic to $G[V_1 \cup V_2]$, it must also be a four-cycle and hence in $G[V_1 \cup V_3]$, there is either an edge from $a_1$ to $c_1$ or from $c_1$ to $a_1$. If there is an edge from $a_1$ to $c_1$, then we have a $3$-cycle $a_1 \rightarrow c_1 \rightarrow b_2 \rightarrow a_1$, which visits each part of $G$ at most once and thus this is a contradiction. Similarly, if there is an edge from $c_1$ to $a_1$, then also we have a $3$-cycle $a_1 \rightarrow b_1 \rightarrow c_1 \rightarrow a_1$, which visits each part of $G$ at most once and thus this is also a contradiction.

We suspect that $\ram(d) \in \mathcal{O}(d)$. We believe that finding better upper bounds on $\ram(d)$ is a natural combinatorial question and  better upper-bounds to $\ram(d)$ imply the existence of better relaxations of EFX allocations.  Therefore investigating better upper bounds on  the rainbow cycle number  is of interest in its own right and we leave this as an interesting open problem.

\subsection{Finding $(1-\varepsilon)$-EFX allocations with high Nash welfare.} Let us recall that \emph{efficiency} is also an important and desirable property of the allocations in Fair Division. The efficiency of an allocation is a measure of the overall welfare the allocation achieves. This is important as an envy-free  allocation could be otherwise unsatisfactory: consider a simple instance with two agents $1$ and $2$ and two goods $g_1$ and $g_2$. Let $v_1(g_1) = v_2(g_2) = 1$ and $v_1(g_2) = v_2(g_1) = 0$. Note that $X_1 \gets \{g_2\}$ and $X_2 \gets \{g_1\}$ is an EFX allocation as each bundle is a singleton and following the removal of a single good results in an empty bundle which is unenvied. However, there is clearly a better EFX allocation, where the individual and the total welfare is better, namely $X_1 \gets \{g_1\}$ and $X_2 \gets \{g_2\}$. 

\emph{Nash welfare} of an allocation $X$, defined as the geometric mean of the valuations of the agents, $(\prod_{i \in [n]} v_i(X_i) )^{\nicefrac{1}{n}}$ is a popular measure of economic efficiency.\footnote{It implies other notions of efficiency like \emph{Pareto-optimality}. An allocation $X = \langle X_1,\dots ,X_n \rangle$ is Pareto-optimal if there is no allocation $Y = \langle Y_1, \dots , Y_n \rangle$ where $v_i(Y_i) \geq v_i(X_i)$ for all $i \in [n]$ and $v_j(Y_j) > v_j(X_j)$ for some $j$.} In fact, when agents have additive valuations, then the allocation with the highest Nash welfare is also EF1 (another popular fairness notion weaker than EFX). Unfortunately, maximizing Nash welfare is APX-hard. However, there have been several approximation algorithms~\cite{ColeG18, AnariGSS17, BKV18} that give a constant factor approximation. The best approximation ratio is $e^{1/e} \approx 1.445$, given  by Barman et al.~\cite{BKV18}.

Similar to the algorithm in~\cite{CKMS20}, we show that with minor modifications to our main algorithm, we can determine an allocation that satisfies the conditions in Theorem~\ref{mainthm1}, and simultaneously achieves a $2e^{1/e} \approx 2.88$ approximation of the Nash welfare, i.e.,  in polynomial time we can find efficient $(1-\varepsilon)$-EFX allocation with sublinear number of unallocated goods. 

\begin{theorem}
	\label{mainthm4}
	For all $\eps \in (0,1/2]$ we can determine a partial $(1-\varepsilon)$-EFX allocation $X$ and a set of unallocated goods $P$ in polynomial time such that  $\abs{P} \le 64 (n / \eps)^{4 / 5}$ and  $\mathit {NW}(X) \geq (1/2.88) \cdot \mathit{NW}(X^*)$, where $X^*$ is the allocation with highest Nash welfare.
\end{theorem}

\subsection{Further Related Work}
Fair division has received significant attention since the seminal work of Steinhaus~\cite{Steinhaus48} in the 1940s. Other than envy-freeness, another fundamental fairness notion is that of \emph{proportionality}. Recall that, in an envy-free allocation, every agent values her own bundle at least as much as she values the bundle of any other agent. However, in a proportional allocation, each agent gets a bundle that is worth $1/n$ times her valuation on the entire set of goods. Since envy-freeness and proportionality cannot always be guaranteed while dividing indivisible goods, various relaxations of the same have been studied. Alongside EFX, another popular relaxation of envy-freeness is \emph{envy-freeness up to one good (EF1)} where no agent envies another agent following the removal of \emph{some} good from the other agent's bundle. While the existence of EFX allocations is open, EF1 allocations are known to exist for any number of agents, even when agents have weakly monotone valuation functions~\cite{LiptonMMS04}. While EF1 and EFX are fairness notions that relax envy-freeness, the most popular notion of fairness that relaxes proportionality for indivisible items is \emph{maximin share} (MMS), which was introduced by Budish~\cite{budish2011combinatorial}. While MMS allocations do not always exist~\cite{KPW18}, but there has been extensive work to come up with approximate MMS allocations~\cite{budish2011combinatorial,BL16,AMNS17,BK17,KPW18,GhodsiHSSY18,JGargMT19,GargT19}. Some works assume ordinal ranking over the goods, as opposed to cardinal values, e.g.,~\cite{AzizGMW15,BramsKK17}.

Alongside fairness, the efficiency of an allocation is also a desirable property. Two common measures of efficiency is that of Pareto-optimality and Nash welfare.  Caragiannis et al.\ \cite{CaragiannisKMP016} showed that any allocation that has the maximum Nash welfare is guaranteed to be Pareto-optimal (efficient) and EF1 (fair).  Barman et al.~\cite{BKV18} give a pseudopolynomial algorithm to find an allocation that is both EF1 and Pareto-optimal. Other works explore relaxations of EFX with high Nash welfare~\cite{CaragiannisGravin19, CKMS20}. 
\medskip

\noindent{\bf The rest of the paper is organized as follows:}
In Section~\ref{tech-overview}, we briefly highlight our main techniques used to prove our main results (Theorem~\ref{mainthm1},~\ref{mainthm2} and~\ref{mainthm3}). Then in Section~\ref{prelim}, we outline the basic concepts, notations and techniques from existing literature on EFX allocations that will be useful to prove our main results. In Sections~\ref{transformation} and~\ref{boundonR}, we give the proofs of Theorem~\ref{mainthm2} and Theorem~\ref{mainthm3} respectively. In Section~\ref{efficiency}, we show how a minor modification of our main algorithm helps us achieve our main result (Theorem~\ref{mainthm1}) with high Nash welfare (efficiency guarantees). Finally, in Section~\ref{limitationoftechnique}, we show why the technique from~\cite{CGM20} does not extend to a setting with four agents with additive valuations.

\section{Our Techniques}
\label{tech-overview}
In this section, we give a brief overview of our key ideas and techniques. We first sketch the key idea that relates the number of unallocated goods to the function  the rainbow cycle number  (Theorem~\ref{mainthm2}) and then we briefly show that $\ram(d)$ is finite. 

\paragraph{Relation between the number of unallocated goods to the rainbow cycle number.} A very crucial concept that is often used while studying relaxations of envy-freeness in discrete fair division is the \emph{envy-graph of an allocation}. Given an allocation $X = \langle X_1, X_2, \dots ,X_n \rangle$, the envy-graph $E_X$ has vertices corresponding to the agents and there is an edge from agent $i$ to agent $j$ in $E_X$ if agent $i$ envies agent $j$ ($v_i(X_i) < v_i(X_j)$). Without loss of generality, one assumes that the envy-graph of an allocation is acyclic: If there is a cycle, then one can shift the bundles along the cycle, thereby giving every agent in the cycle a strictly better bundle and the other agents retain their previous bundle. Such a procedure reduces the number of edges in the envy-graph, and one can continue this until $E_X$ is cycle-free. 

Most of the algorithms that have been used to prove the existence of relaxations of EFX allocations~\cite{CKMS20, CGM20, TimPlaut18} maintain a \emph{relaxed EFX allocation}\footnote{$(1-\varepsilon)$-EFX allocation in~\cite{CKMS20, CGM20} and $1/2$-EFX allocation in~\cite{TimPlaut18}.} $X$ on the set of allocated goods and as long as the envy-graph $E_X$ and the set of unallocated goods satisfy some ``properties", they determine another relaxed EFX allocation $X'$, in which $\phi(X') \geq \phi(X) + \delta$ for some $\delta \geq 1$, where $\phi$ is an integral upper-bounded function. In that case, we say that the relaxed EFX allocation $X'$ \emph{dominates} the relaxed EFX allocation $X$. Since $\phi$ is integral and upper-bounded, such a procedure will finally converge to a relaxed EFX allocation where the envy graph $E_X$ and the unallocated goods will not satisfy the said properties and this will be the final allocation of the algorithms.

We now highlight another crucial concept used in these algorithms. The envy-graph $E_X$ does not provide any information on an agent's valuations of the bundles formed by adding unallocated goods to the current bundles of the allocation. This information is crucial when we want to create another {dominating} relaxed EFX allocation by allocating some of the unallocated goods and unallocating some of the already allocated goods. The algorithms in~\cite{CKMS20,CGM20} make use of this information through other concepts. For instance \cite{CKMS20, CGM20} define \emph{champions}\footnote{They are called ``most envious agents" in~\cite{CKMS20}} and \emph{champion graphs}. Given an allocation $X$ and an unallocated good $g$, we say that an agent $i$ is a champion for agent $j$ w.r.t $g$ if there is a set $S \subseteq X_j \cup \{g\}$ such that $v_i(X_i) < (1-\varepsilon) \cdot v_i(S)$ and no agent (including $i$ and $j$) envies $S$ up to a factor of $(1-\varepsilon)$, following the removal of a single good, i.e, for all $\ell \in [n]$, we have $(1-\varepsilon) \cdot v_{\ell}(S \setminus \{h\}) \leq v_{\ell}(X_{\ell})$ for all $h \in S$.\footnote{Since we are dealing with $(1-\varepsilon)$-EFX allocations and not EFX allocations, we have changed the definition of champions and champion graphs appropriately.~\cite{CGM20, CKMS20} also use this definition in their algorithms as the polynomial time algorithms also deal with $(1-\varepsilon)$-EFX allocations.} A champion graph w.r.t an unallocated good $g$ has vertices corresponding to the agents (similar to the envy graph) and there is an edge from agent $i$ to agent $j$ if agent $i$ champions agent $j$ w.r.t. $g$. Depending on the configuration of the envy-graph and the champion graphs (one for each unallocated good), the current $(1-\varepsilon)$-EFX allocation $X$ is transformed into another $(1-\varepsilon)$-EFX allocation $X'$ such that $X'$ dominates $X$.  However, when the number of agents are large, there are several different possible configurations of the champion graphs and the envy-graph and it is very hard and tedious to come up with better update rules. In this paper, we introduce the notion of a \emph{group champion graph} which is significantly more insightful and well structured than the champion graphs.

Given a $(1-\varepsilon)$-EFX allocation $X$ and a set of unallocated goods $M'$, we define the group champion graph. To this end, for each agent $a \in [n]$, we assign a unique source $s(a)$ in $E_X$ such that $a$ is reachable from $s(a)$ in $E_X$ (if there are multiple sources from which $a$ is reachable in $E_X$, then pick one source arbitrarily). The group champion graph of $M'$ is a $\abs{M'}$-partite graph $G = (\cup_{g \in M'} V_g, E)$, in which each part $V_g$ contains a copy of the assigned sources of all the agents that find $g$ ``valuable"; an agent $a$ finds $g$ valuable if $v_a(\{g\}) > \varepsilon \cdot v_a(X_a)$. There is an edge from vertex $s(a)$ in $V_g$ to $s(a')$ in $V_h$ if and only if $a$ champions $s(a')$ w.r.t $g$ (see Figure~\ref{group-champion-graph} for an illustration). At a high level, the group champion graph encodes the most relevant information from all the champion graphs. We make this point more explicit by briefly explaining how group champion graphs help us prove Theorem~\ref{mainthm2}.   

\begin{figure}
	\centering 
	\begin{tikzpicture}
	[
	agent/.style={circle, draw=green!60, fill=green!5, very thick},
	good/.style={circle, draw=red!60, fill=red!5, very thick, minimum size=1pt},
	]
	%Parts
	\draw[black, very thick] (7-0.5,0.5) rectangle (7+0.5,-2.5);
	\draw[black, very thick] (7-0.5+2,0.5) rectangle (7+0.5+2,-2.5);
	\node at (7,-2.75) {$V_{g_a}$};
	\node at (9,-2.75) {$V_{g_b}$};
	
	%Vertices for config1
	\node[agent]      (a1) at (0,0)      {$\scriptstyle{a_1}$};
	\node[agent]      (a2) at (0,-2)      {$\scriptstyle{a_2}$};
	\node[agent]      (a3) at (2,0)     {$\scriptstyle{a_3}$};
	\node[agent]      (a4) at (2,-2)     {$\scriptstyle{a_4}$};
	\node[agent]      (b1) at (4,0)      {$\scriptstyle{b_1}$};
	\node[agent]      (b2) at (4,-2)      {$\scriptstyle{b_2}$};
	
	%Nodes in parts
	\node[agent]      (a11) at (7, -0.5)      {$\scriptstyle{a_1}$};
	\node[agent]      (a31) at (7,-1.5)      {$\scriptstyle{a_3}$};
	\node[agent]      (b11) at (9, -1)      {$\scriptstyle{b_1}$};

	%Edges in E_X
	\draw[->,blue,thick] (a1) -- (a2);
	\draw[->,blue,thick] (a3) -- (a4);
	\draw[->,blue,thick] (b1) -- (b2);

	%Edges in Champion Graphs
	\node at (2,-3.25) {\small{$a_2$ champions all agents w.r.t $g_a$.}};
	\node at (2,-4) {\small{$b_2$ champions all agents w.r.t $g_b$.}};

	%Edges in group champion graph
	\draw[->,black,thick] (a11) -- (b11);
	\draw[->,black,thick] (b11) -- (a31);	
\end{tikzpicture}	
	\caption{Illustration of a group champion graph. We have an instance with six agents $\cup_{i \in [4]} a_i$ and $\cup_{i \in [2]} b_i$ and two unallocated goods, namely $g_a$ and $g_b$. The agents $\cup_{i \in [4]}a_i$ find $g_a$ valuable and the agents $\cup_{i \in [2]} b_i$ find $g_b$ valuable. The envy graph $E_X$ of the instance is shown on the  left side. $E_X$ shows that $s(a_2) = a_1$, $s(a_4) = a_3$, and $s(b_2) = b_1$.  Also, we have that agent $a_2$ champions all the agents w.r.t $g_a$ and $b_2$ champions all the agents w.r.t $g_b$. The group champion graph (right) has two parts, $V_{g_a}$ corresponding to $g_a$ and $V_{g_b}$ corresponding to $g_b$. $V_{g_a}$ contains the sources of all the agents that find $g_a$ valuable, namely $a_1$ and $a_3$. Similarly, $V_{g_b}$ contains $b_1$. There is an edge from $a_1$ to $b_1$ as $a_2$ (which is reachable from $a_1$ in $E_X$) champions $b_1$ w.r.t to $g_a$. Similarly, there is an edge from $b_1$ to $a_3$ as $b_2$ (which is reachable from $b_1$ in $E_X$) champions $a_3$ w.r.t to $g_b$. }
	\label{group-champion-graph}
\end{figure}
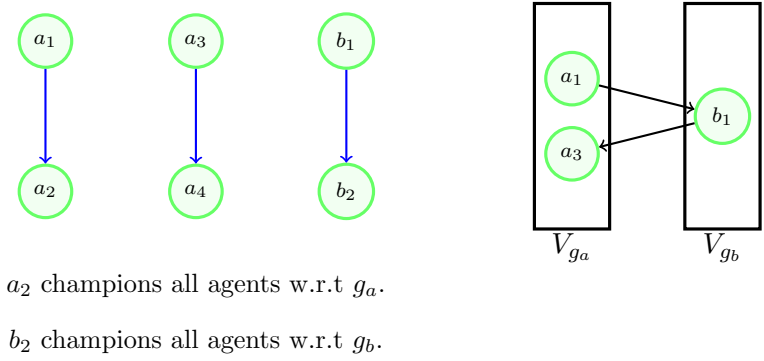

We first observe that if there is an unallocated good $g$ and an agent $i$ such that the other agents do not envy $(X_i \cup \{g\}) \setminus g'$ up to a factor of $(1-\varepsilon)$ for all $g' \in X_i \cup \{g\}$, then we allocate $g$ to $i$. Thus, we assume that for each unallocated good $g$ and each agent $i$, there is an agent $j$ that envies $(X_i \cup \{g\}) \setminus g'$ up to a factor of $(1-\varepsilon)$ for some $g' \in X_i \cup \{g\}$. In particular, this implies that every unallocated good is valuable to some agent, because if there is a good $g$ which is not valuable to any agent, i.e., $v_i(\{g\}) \leq \varepsilon \cdot v_i(X_i)$ for all $i \in [n]$, then we can simply allocate $g$ to a source $s$ in $E_X$ as no agent will envy the bundle $X_s \cup \{g\}$ up a factor of $(1-\varepsilon)$: for all $i \in [n]$, we have that $v_i(X_i) \geq v_i(X_s)$ (as $s$ is unenvied) and $\varepsilon \cdot v_i(X_i)  \geq v_i(\{g\})$, implying that $(1+\varepsilon) v_i(X_i) \geq v_i(X_s \cup \{g\})$, further implying that $v_i(X_i) \geq (1-\varepsilon) \cdot v_i(X_s \cup \{g\})$. Now we classify the set of unallocated goods into two categories depending on how many agents find them valuable: We fix an integer $d < n$ and define ``high-demand goods" and ``low-demand goods". A high-demand good is valuable to more than $d$ agents and a low-demand good is valuable to at most $d$ agents. We show in Section~\ref{transformation}, that if the number of high-demand goods is more than $2n/(\varepsilon \cdot d)$, then we can determine a dominating $(1-\varepsilon)$-EFX allocation from the existing $(1-\varepsilon)$-EFX allocation. Thus, we may assume that the number of high-demand goods is at most $2n/(\varepsilon d)$. We now bound the number of low-demand goods. Let $M''$ be the set of low-demand goods. We construct the group champion graph $G = (\cup_{g \in M''} V_g, E)$ of $M''$ in which part $V_g$ contains the assigned sources of the agents that find $g$ valuable. Note that for all $g \in M''$, $g$ is not valuable to more than $d$ agents. Thus $\abs{V_g} \leq d$ for all $g \in M''$. Now, consider any two parts $V_g$ and $V_h$ in $G$. By our assumption, for all $a \in V_h$, there is an agent that envies $(X_a \cup \{g\}) \setminus g'$ up to a factor of $(1-\varepsilon)$ for some good $g' \in X_a \cup \{g\}$, implying that for each $a$ in $V_h$, there are agents that champion $a$ w.r.t $g$. Since $a$ is a source in $E_X$, it is intuitive that the agents that champion $a$ w.r.t. $g$, must find $g$ valuable. Therefore, for all $a \in V_h$, there is a source $s(a') \in V_g$, where $a'$ champions $s(a)$ w.r.t $g$. Thus, every vertex in $V_h$ has an incoming edge from a vertex in $V_g$. \emph{In Section~\ref{transformation}, we show that whenever $G$ has a cycle that visits each part at most once, then we can determine a $(1-\varepsilon)$-EFX allocation that dominates $X$.} Therefore, we can assume that $G$ has no cycle that visits each part at most once. Since $G$ is a $\abs{M''}$-partite graph that satisfies the conditions in Definition~\ref{Rdef}, we have that the number of low-demand goods is $\abs{M''} \leq \ram(d)$. Therefore, the total number of unallocated goods is $\mathcal{O}(\mathit{max}(2n/(\varepsilon \cdot d), \ram(d)))$. By choosing the appropriate value for $d$, we arrive at the statement of Theorem~\ref{mainthm2}.

We now elaborate that $\ram(d)$ is indeed  upper-bounded, which then establishes the existence of $(1-\varepsilon)$-EFX allocations with sublinear number of unallocated goods.

\paragraph{Upper bounds on the rainbow cycle number.} 
We briefly show that for any $d \in \mathbb{N}$, $\ram(d)$ is finite. Consider a $k$-partite graph $G = (\cup_{i \in [k]} V_i, E)$ in Definition~\ref{Rdef}. For all $i \in [k]$, let $V_i = \{(i,1), (i,2), \dots, (i,\abs{V_i}) \}$. For all $i<j$ and $i'<j'$, we say that the directed bipartite graphs $G[V_i \cup V_j]$ and $G[V_{i'} \cup V_{j'}]$ have the same configuration if and only if for each directed edge from vertex $(i,a)$ to $(j,b)$ (and equivalently from $(j,b')$ to $(i,a')$) in $G[V_i \cup V_j]$, there is an edge from $(i',a)$ to $(j',b)$ (and equivalently from $(j',b')$ to $(i',a')$) in $G[V_{i'} \cup V_{j'}]$ and vice-versa. We first show that if there are $4d$ parts in $G$, say w.l.o.g. $V_1, V_2, \dots , V_{4d}$, such that  the induced directed bipartite graph $G[V_i \cup V_j]$ has the same configuration for all $1 \leq i<j \leq 4d$, then there exists a cycle in $G$ that visits each part at most once.

Consider the parts $V_1$ and $V_2$, and the induced directed bipartite graph $G[V_1 \cup V_2]$. Since every vertex in one part has an incoming edge from a vertex in the other part, $G[V_1 \cup V_2]$ is cyclic. Let the simple cycle be $C = (1,i_1) \rightarrow (2,i_2) \rightarrow (1,i_3) \rightarrow \dots \rightarrow (2,i_{2\beta}) \rightarrow (1,i_1)$ for some $\beta \leq d$. 
Since all the induced bipartite graphs $G[V_i \cup V_j]$ have the same configuration for all $1 \leq i<j \leq 4d$, we can claim that for all $\ell \in [\beta]$, for each edge $(1,i_{2\ell-1}) \rightarrow (2,i_{2\ell})$ in $C$, there is an edge  from $(2\ell-1,i_{2\ell-1})$ to  $(4d-\ell, i_{2\ell})$ in $G[V_{2\ell-1}, V_{4d-\ell}]$ (note that $2\ell -1 < 4d-\ell$ as $\ell \leq \beta \leq d$).  Similarly for all $\ell \in [\beta]$, for each edge $(2,i_{2\ell}) \rightarrow (1,i_{2\ell+1})$ in $C$ ($2\beta +1$ is to interpreted as $1$), there is an edge  from $(4d - \ell,i_{2\ell})$ to  $(2\ell+1, i_{2\ell+1})$ in $G[V_{2\ell+1}, V_{4d -\ell}]$ (again, note that $ 2\ell+1 < 4d -\ell$ as $\ell \leq \beta \leq d$). This implies that there is a cycle $C' = (1,i_1) \rightarrow (4d-1,i_2) \rightarrow (3,i_3) \rightarrow (4d-2,i_4) \rightarrow \dots \rightarrow  (4d - \beta, i_{2\beta}) \rightarrow (1,i_1)$ in $G$. Clearly, $C$ visits each part of $G$ at most once.  Therefore, there cannot be $4d$ parts in $G$ such that  the induced directed bipartite graph $G[V_i \cup V_j]$ has the same configuration for all $1 \leq i<j \leq 4d$.   

We now rephrase the question about an upper bound on $\ram(d)$. Let $\mathcal{D}$ be the set of all configurations of a directed bipartite graph, where the number of vertices in each part is at most $d$ and every vertex has an incoming edge. We treat $\mathcal{ D}$ as a set of \emph{colors} and note that $\abs{\mathcal{D}} \in 2^{\mathcal{O}(d^2)}$.  Now consider a complete graph $K_k$ with vertex set $[k]$, where the vertex $\ell \in [k]$ corresponds to part $V_{\ell}$ in $G$. For all $1 \leq i<j \leq k$, we color/label the edge $(i,j)$ in $K_k$ with a color from $\mathcal{D}$. The color on the edge $(i,j)$  corresponds to the configuration of the directed bipartite graph $G[V_i \cup V_j]$. Clearly, $\ram(d)$ must be strictly smaller than the largest $k$ such that every coloring of the edges of $K_k$ with colors from $\mathcal{D}$ contains a monochromatic clique of size $4d$.  This value of $k$ corresponds to the \emph{(multicolor) Ramsey number}~\cite{Diestel-graph-theory} $\mathcal{R}(n_1,n_2, \dots n_{\abs{\mathcal{D}}})$ in which $n_i = 4d$ for all $i \in [\abs{\mathcal{D}}]$.  
This number is finite and the current best known upper bounds on it are exponential in $|\mathcal{ D}|$ and $d$~\cite{erdos1935combinatorial,lefmann1987note,Diestel-graph-theory,conlon2021lower}. Therefore, $\ram(d)$ is also bounded. However, this upper-bound is very large and only provides a weak version of Theorem~\ref{mainthm1}. This necessitates the study of finding ``good" upper bounds on $\ram(d)$; in particular, upper bounds that are polynomial in $d$. We address this in Section~\ref{boundonR} by showing that $\ram(d) \in \mathcal{O} (d^4)$.

\section{Preliminaries and Tools} \label{prelim}
A fair division instance is given by the three tuple $\langle [n],M ,\mathcal{V} \rangle$, where $[n]$ is the set of agents, $M$ is the set of indivisible goods, and $\mathcal{V} = \{v_1(),v_2(), \dots, v_n() \}$ where each $v_i \colon 2^M \rightarrow \mathbb{R}_{\geq 0}$ denotes the valuation function of agent $i$. We assume that agents have \emph{additive valuations}, i.e, for all $i \in [n]$, we have $v_i(S) = \sum_{g \in S} v_i(\{g\})$ for all $S \subseteq M$. For the ease of notation, we write $v_i(g)$ instead of $v_i(\{g\})$ and similarly $v_i(S \cup g)$ for $v_i(S \cup \{g\})$. We assume that $v_i(g)$ can be accessed in constant time for any $i$ and $g$. For a fixed $0< \epsilon < 1$, we say that an agent $i$
\begin{itemize}
	\item \emph{envies} a set $S$ of goods if $v_i(X_i) < v_i(S)$, 
	\item \emph{heavily envies} a set $S$ of goods if $v_i(X_i) < (1 - \varepsilon) v_i(S)$, 
	\item \emph{strongly envies} a set $S$ of goods if it heavily envies a proper subset of $S$, and
	\item is a \emph{most envious agent} for a set $S$ of goods if there exists a subset $Z \subseteq S$ such that $i$ heavily envies $Z$ and no agent strongly envies $Z$. The pair $(t,Z)$ is called a \emph{most-envious-agent-witness pair} for $S$. 
\end{itemize}
An agent envies (heavily envies, strongly envies) an agent $j$ if it has these feelings for the set $X_j$.  Clearly, strong envy implies heavy envy implies envy. An allocation $X'$ \emph{strongly Pareto-dominates} an allocation $X$, or equivalently  $X' \pd X$, if and only if $v_i(X'_i) \geq v_i(X_i)$ for all $i \in [n]$ and for some agent $i' \in [n]$ we have $(1-\varepsilon) \cdot v_{i'}(X'_{i'}) \geq v_{i'}(X_{i'})$.

At a high level, our algorithm is similar to previous algorithms used to prove the existence of relaxations of EFX allocations~\cite{CKMS20, TimPlaut18, CGM20}. Our algorithm always maintains a $(1-\varepsilon)$-EFX allocation on the set of allocated goods and as long as the current allocation and the set of unallocated goods $P$ satisfies ``some properties'', it determines another $(1-\varepsilon)$-EFX allocation that strongly Pareto-dominates the previous $(1-\varepsilon)$-EFX allocation. Since the valuation of an agent for the entire good set is bounded, this procedure will eventually converge to a $(1-\varepsilon)$-EFX allocation, where the current allocation and the set of unallocated goods do not satisfy these properties. The bulk of the effort goes into determining the right properties under which one can  come up with update rules that transform one $(1-\varepsilon)$-EFX allocation into a ``better'' $(1-\varepsilon)$-EFX allocation. We briefly recollect the update rules used in~\cite{LiptonMMS04} and~\cite{CKMS20}.

\paragraph{Envy cycle elimination~\cite{LiptonMMS04}.} The \emph{envy-graph} $E_X$ of an $(1-\varepsilon)$-EFX allocation $X$ has the agents as its vertex set  and there is an edge from vertex $i$ to vertex $j$ in $E_X$ if agent $i$ envies agent $j$, i.e., $v_i(X_i) < v_i(X_j)$. The paper~\cite{LiptonMMS04} shows that whenever $E_X$ has a cycle, then one can determine another $(1-\varepsilon)$-EFX allocation $X'$ in which no agent has a worse bundle and $E_{X'}$ is acyclic. Formally,

\begin{lemma}[\cite{LiptonMMS04}]
	\label{envycycleelimination}
	Consider a $(1- \varepsilon)$-EFX allocation $X$. If there is a cycle in $E_X$, then in polynomial time, we can determine a  $(1-\varepsilon)$-EFX allocation $X'$ such that $v_i(X'_i) \geq v_i(X_i)$ for all $i \in [n]$, and  $E_{X'}$ is acyclic.\footnote{Let $C$ be an envy cycle. For each edge $(i,j)$ of the cycle one assigns in $X'$ the bundle $X_j$ to $i$. One continues in this way as long as there is a cycle in the envy graph.}
\end{lemma}

\paragraph{Update rules in~\cite{CKMS20}.} These rules\footnote{We modify the update rules in~\cite{CKMS20} slightly, as we are dealing with $(1-\varepsilon)$-EFX allocations and not EFX allocations} are more involved and make essential use of the concept of \emph{a most envious agent}. 

\begin{observation}
	\label{mostenviousagent}
	Consider an allocation $X$ and a set $S \subseteq M$. If there is an agent that heavily envies the bundle $S$, then we can determine a most-envious-agent-witness pair $(t,Z)$ for $S$ in $\mathcal{O}(n \cdot \abs{S}^2)$ time. If there is an agent that strongly envies $S$ then $t$ strongly envies $S$. 
\end{observation}

\begin{proof}
	Let $i$ be an agent that heavily envies $S$.  We construct a sequence $(t_{\ell}, Z_{\ell})$ as follows: initially we set $t_1$ to $i$ and $Z_1$ to $S$. Assume that $(t_{\ell-1},Z_{\ell-1})$ is defined. If no agent (including $t_{\ell-1}$) strongly envies $Z_{\ell-1}$, then we stop. Otherwise let $i'$ be an agent such that $v_{i'}(X_{i'}) < (1-\varepsilon) \cdot v_{i'}(Z_{\ell-1} \setminus \{g \})$ for some $g \in Z_{\ell-1}$. We set $t_{\ell}$ to $i'$ and $Z_{\ell}$ to $Z_{\ell-1} \setminus \{ g \}$ and continue. We will eventually stop, as with every next pair in the sequence, the size of the set $Z_{\ell}$ decreases by one. Say we stop at $\ell^*$. Then, we have an agent $t_{\ell^*}$ that heavily envies the subset  $Z_{\ell^*}$ of $S$. Moreover, no agent strongly envies $Z_{\ell^*}$. Thus $(t_{\ell^*}, Z_{\ell^*})$ is a most-envious-agent-witness pair.
	
	If there is an agent that strongly envies $S$ then $\ell \ge 1$ and hence $t_{\ell^*}$ heavily envies a proper subset of $S$. Thus $t_{\ell^*}$ strongly envies $S$. 
	
	It is clear that we can determine the pair in $\mathcal{O}(n \cdot \abs{S}^2)$ time: the maximum length of the sequence constructed is $\abs{S} + 1$ as the size of the set $Z_{\ell} = \abs{S} +1 - \ell$. We need time $\O(n \abs{S})$ to determine $v_i(S)$ for all $i$ and can update any such value in time $\O(1)$ after the removal of an element. For each value of $\ell$, it takes $\mathcal{O}(n \cdot \abs{Z_{\ell}}) =  \mathcal{O}(n \cdot \abs{S})$ time to find $(t_{\ell+1}, Z_{\ell+1})$. Thus the total time needed is  $\mathcal{O}(n \cdot \abs{S}^2)$. 
\end{proof}

For an allocation $X$ and set $S$ of goods that is heavily envied by some agent, let $(t,Z)$ be the pair returned by the procedure in Observation~\ref{mostenviousagent}. We call $t$ \emph{the champion} of $S$ and $Z$ the corresponding witness. 

We now state the update rules. The first rule is the simplest. It is applicable whenever we can allocate an unallocated good to an unenvied agent (a source in $E_X$), without creating any strong envy. In this case, we simply allocate this good to the corresponding source. This creates another $(1-\varepsilon)$-EFX allocation where no agent gets a worse bundle and the number of unallocated goods decreases. 

\begin{lemma}[$U_1$~\cite{CKMS20}]
	\label{U1}
	Consider a $(1-\varepsilon)$-EFX allocation $X$. If there is a source $s$ in $E_X$ and an unallocated good $g$ such that no agent strongly envies $X_s \cup g$, then $X' = \langle X_1, X_2, \dots, X_s \cup g, \dots , X_n \rangle$ is a $(1-\varepsilon)$-EFX allocation and $v_i(X'_i) \geq v_i(X_i)$ for all $i \in [n]$.
\end{lemma}

Note that there can be at most $m$ consecutive applications of this rule as the number of unallocated goods decreases by one every time we apply this update rule. The remaining rules are applicable whenever, there are either ``valuable" goods unallocated or if ``too many" goods are unallocated. We state the second update rule, which is applicable if there is any agent that heavily envies the set of unallocated goods.  Formally,

\begin{lemma}[$U_2$~\cite{CKMS20}]
	\label{U2}
	Consider a $(1-\varepsilon)$-EFX allocation $X$ and let $P$ be the set of unallocated goods. If there is an agent $i \in [n]$ such that heavily envies $P$, then in polynomial time, we can determine\footnote{Let $t$ be the champion of $P$ and $Z$ be the corresponding witness. In $X'$, one assigns $Z$ to $t$ and changes the pool to $X_t \cup (P \setminus Z)$.} a $(1-\varepsilon)$-EFX allocation $X' \pd X$.
\end{lemma}

The third update rule is a refinement of envy-cycle elimination. In~\cite{CKMS20} it was shown that it is applicable whenever the number of unallocated goods is at least the number of agents.

\begin{lemma}[$U_3$~\cite{CKMS20}]
	\label{U3}
	Consider a $(1-\varepsilon)$-EFX allocation $X$. If there exists a set of sources $s_1, s_2, \dots s_{\ell}$ in $E_X$, a set of unallocated goods $g_1, g_2, \dots , g_{\ell}$, and a set of agents $t_1, t_2, \dots, t_{\ell}$, such that each $t_i$ is reachable from $s_i$ in $E_X$ and $t_{i}$ is the champion of $X_{s_{i+1}} \cup g_{i+1}$ (indices are modulo $\ell$), then  in polynomial time, we can determine\footnote{Let $Z_{i+1} \subseteq X_{s_{i+1}} \cup g_{i+1}$ be the witness corresponding to $t_i$. One then essentially proceeds as in cycle elimination. For each $i$ one assigns $Z_{i+1}$ to $t_i$ and to each agent on the path from $s_i$ to $t_i$ except for $t_i$ one assigns the bundle owned by the successor on the path.} a $(1-\varepsilon)$-EFX allocation $X' \pd X$.
\end{lemma}

\section{Relating the Number of Unallocated Goods to the Rainbow Cycle Number}
\label{transformation}
In this section, we give the proof of Theorem~\ref{mainthm2}, i.e, we show how any upper bound on $\ram(d)$ allows us to obtain a $(1-\varepsilon)$-EFX with sublinear many goods unallocated. More precisely, we show that given a $(1-\varepsilon)$-EFX allocation $X$, if $E_X$ is acyclic, and the update rules $U_1$ and $U_2$ are not applicable, and the number of unallocated goods is larger than $4n/ (\varepsilon \cdot h^{-1}(2n/ \varepsilon))$, then rule $U_3$ is applicable. Therefore,  for most of this section, we proceed under the assumption
\begin{equation}\tag{*} \text{$E_X$ is acyclic and the update rules $U_1$(Lemma~\ref{U1}) and $U_2$ (Lemma~\ref{U2}) are not applicable.}\end{equation}

We start with some definitions. Given a partial allocation $X$, we call an unallocated good $g$ \emph{valuable} to an agent $i$ if $v_i(g) > \varepsilon \cdot v_i(X_i)$.  We first make an observation about the agents that could potentially strongly envy $X_s \cup g$, where $s$ is a source in $E_X$ and $g$ is an unallocated good.  

\begin{observation}
	\label{onlyvaluablenevy}
	Consider an unallocated good $g$ and any source $s$ in $E_X$. If agent $i$ heavily envies $X_s \cup g$, then  $g$ is valuable to agent $i$.
\end{observation}

\begin{proof} We have $v_i(X_s) \le v_i(X_i)$ since $s$ is a source of $E_X$ and $v_i(X_i) < (1 - \varepsilon) v_i(X_s \cup g)$ since $i$ heavily envies $X_s \cup g$. Thus $v_i(X_i) < (1 - \varepsilon) (v_i(X_i) + v_i(g))$ and hence
	$(1 - \varepsilon) v_i(g) > \varepsilon v_i(X_i)$. 
\end{proof}

Note that under assumption (*) for each unallocated good $g$, and each source $s$ in the envy-graph, there is an agent that strongly envies $X_s \cup g$ (since the conditions of the update rule $U_1$ (Lemma~\ref{U1}) are not satisfied). Thus, each unallocated good is valuable to some agent. Now, we make a classification of the unallocated goods based on the number of agents that find them valuable. To be precise, given an allocation $X$, we classify the unallocated goods into two categories: \emph{high-demand goods} $H_X$ and \emph{low-demand goods} $L_X$. A good $g$ belongs to $H_X$, if it is valuable to at least $d+1$ agents and to $L_X$ if it is valuable to at most $d$ agents. We will choose the exact value of $d$ later (right now, just think of it as any integer less than $n$). Observe that the set of unallocated goods $P  = H_X \cup L_X$. To prove our claim, it suffices to show that when $\abs{H_X} + \abs{L_X} > 4n / (\varepsilon \cdot h^{-1}(2n / \varepsilon))$, the rule $U_3$ is applicable. To this end, we first make a simple observation about $\abs{H_X}$.

\begin{observation}
	\label{V_Xbound}
	Under assumption (*), we have $\abs{H_X} < 2n / (\eps \cdot d)$. 
\end{observation}

\begin{proof}
	For each good $g \in H_X$, let $\eta_g$ be the number of agents that find $g$ valuable. By definition of $H_X$, we have that $\eta_g > d$ and hence $\sum_g \eta_g > \abs{H_X} d$. We next upper bound $\sum_g \eta_g$ by $n \cdot (2/\eps)$ by showing that at most $2/\eps$ unallocated goods are valuable to any agent.
	
	Consider any agent $i$. By assumption (*) rule $U_2$ is not applicable and hence the value of the unallocated goods to $i$ is at most $1/(1 - \eps) v_i(X_i)$. This is at most $2v_i(X_i)$ since $\eps \le 1/2$. Any valuable good has value at least $\eps v_i(X_i)$ for $i$. Thus the number of unallocated goods valuable to $i$ is at most $2/\eps$. 
\end{proof}

We next bound $\abs{L_X}$. In particular, we show that $\abs{L_X} \leq \ram(d)$. To this end, we introduce the notion of \emph{group champion graph} $G$.

\paragraph{Group champion graph.} To each agent $a$, we assign a source $s(a)$, such that $a$ is reachable from $s(a)$ in the envy-graph $E_X$. Recall that we operating under assumption (*) and hence $E_X$ is acyclic. If $a$ is reachable from multiple sources, we pick $s(a)$ arbitrarily from these sources. Let $k := \abs{L_X}$. For each $g \in L_X$, let $Q_g$ be the set of all agents that find $g$ valuable. By definition of $L_X$, we have $\abs{Q_g} \leq d$ for all $g \in L_X$.  We now define a $k$-partite graph $G = (\cup_{g \in L_X} V_g, E)$, in which the part $V_g$ corresponding to $g$ consists of copies of the sources assigned to the agents in $Q_g$, formally,  $V_g = \set{ (g,s(a))}{a \in Q_g}$. For any goods $g$ and $h$ and agents $a \in Q_g$ and $b \in Q_h$, there is an edge from $(g,s(a))$ in $V_g$ to $(h,s(b))$ in $V_h$ if and only if $a$ is the champion of $X_{s(b)} \cup g$. We now make an observation about the set of edges between $V_g$ and $V_h$ in $G$ for any $g,h \in L_X$.

\begin{observation}
	\label{edgesinG} Under assumption (*): 
	Consider any $g,h \in L_X$. Then each vertex in $V_h$, has an incoming edge from a vertex in $V_g$.
\end{observation}

\begin{proof}
	Consider any vertex $(h,s(b)) \in V_h$. By assumption (*), there is an agent that strongly envies the bundle $X_{s(b)} \cup g$. Otherwise, rule $U_1$ would be applicable. By Observation~\ref{onlyvaluablenevy}, all agents that strongly envy $X_{s(b)} \cup g$, consider $g$ valuable and hence belong to $Q_g$. Let $a$ be the champion of $X_{s(b)} \cup g$. By Observation~\ref{mostenviousagent}, $a$ strongly envies $X_{s(b)} \cup g$ and hence belongs to $Q_g$. Thus there is an edge from $(g,s(a))$ in $V_g$ to $(h,s(b))$ in $V_h$ (by the construction of $G$).
\end{proof}

Now we claim that the existence of a cycle that visits each part of $G$ at most once, would imply the existence of a $(1-\varepsilon)$-EFX allocation that Pareto-dominates the existing $(1-\varepsilon)$-EFX allocation. 

\begin{lemma}
	\label{cycleinG}
	Given a cycle $C$ in $G$ that contains at most one vertex from each $V_g$, for all $g \in L_X$, we can determine a $(1-\varepsilon)$-EFX allocation $X' \pd X$ in polynomial time. 
\end{lemma}

\begin{proof}
	Let $C = (g_{i+1},s_i) \rightarrow (g_{i+2},s_{i+1}) \rightarrow \dots \rightarrow (g_{j+1},s_j) \rightarrow (g_{i+1},s_i)$ be a cycle in $G$ that visits each part at most once. It will become clear below, why we index the $g$'s starting at $i+1$. Consider the sequence $s_i, s_{i+1}, \dots, s_j$. If all the sources in this sequence are not distinct, there exists a contiguous subsequence $s_{i'}, s_{i'+1}, \dots , s_{j'}$ where all the sources are distinct and $s_{j'+1} = s_{i'}$ with $i \leq i' < j' \leq j$ (index $j+1$ is to be interpreted as $i$). 	
	
	We now work with the sequence $s_{i'}, s_{i'+1}, \dots , s_{j'}$ where all the sources are distinct and $s_{j'+1} = s_{i'}$. For all $\ell \in [i'+1,j'+1]$, the existence of the edge $(g_{\ell},s_{\ell-1}) \rightarrow (g_{\ell+1} ,s_{\ell})$ implies the existence of an agent $t_{\ell-1}$ such that $t_{\ell - 1}$ is the champion of $X_{s_{\ell}} \cup g_{\ell}$ and $s(t_{\ell - 1}) = s_{\ell - 1}$, i.e., $t_{\ell-1}$ is reachable from $s_{\ell-1}$ in $E_X$.  Since the sources $s_{i'}, s_{i'+1}, \dots , s_{j'}$ are distinct, the agents $a_{i'}, a_{i'+1}, \dots, a_{j'}$ are also distinct (as each agent has a unique source assigned). Therefore, we have distinct sources $s_{i'}, \dots , s_{j'}$ in $E_X$, distinct goods $g_{j'+1}, g_{i'+1}, \dots , g_{j'}$ and distinct agents $t_{i'}, \dots t_{j'}$  that satisfy the conditions under which the update rule $U_3$ (Lemma~\ref{U3}) is applicable. By applying $U_3$ we can get a $(1-\varepsilon)$-EFX allocation $X' \pd X$. 
\end{proof} 

With Lemma~\ref{cycleinG}, we are now ready to give an upper bound on $\abs{L_X}$. Observe that $\abs{L_X}$ equals the number of parts  in $G$. Now the question is how many parts  can $G$ have such that it does not admit a cycle that visits each part at most once. This is where we upper bound $\abs{L_X}$ with the rainbow cycle number.

\begin{lemma}
	\label{L_Xbound}
	Consider a $(1-\varepsilon)$-EFX allocation $X$. If $\abs{L_X} >  \ram(d)$, there is a $(1-\varepsilon)$-EFX allocation $X'\pd X$. %If $\abs{L_X} > d^4 + d$, we can find such an $X'$ in polynomial time. 
\end{lemma}

\begin{proof}
	Recall that $\abs{L_X} = k$, where $k$ is the number of parts  in $G$. Note that each part of $G$ corresponds to the sources assigned to the  agents that find a particular good in $L_X$ valuable ($Q_g$ for some $g \in L_X$). By definition of $L_X$, there are at most $d$ agents that find a good in $L_X$ valuable. Thus each part has at most $d$ vertices. Again, by Observation~\ref{edgesinG}, between any two parts  $V_g$ and $V_h$ of $G$, each vertex in $V_h$ has an incoming edge from a vertex in $V_g$. Therefore, by Definition~\ref{Rdef}, we have that if $k > \ram(d)$, then there exists a cycle $C$ in $G$ that visits each part at most once.
	%Moreover, if $k > d^4 + d$ we can determine such a cycle $C$ in $G$ in polynomial time. 
	Once we have $C$, by Lemma~\ref{cycleinG}, we  can determine a $(1-\varepsilon)$-EFX allocation $X' \pd X$.
\end{proof}

Given a $(1-\varepsilon)$-EFX allocation $X$ such that $\abs{L_X} > \ram(d)$, Lemma~\ref{L_Xbound} only gives the existence of a $(1-\varepsilon)$-EFX allocation $X' \pd X$. However, to determine $X'$ in polynomial time, one needs to find a cycle $C$ in $G$ which visits each part at most once when $\abs{L_X} > \ram(d)$, in polynomial time. Let us remark that this is a non-trivial problem in general, reminiscent of the well-known {\sc $k$-Path} and {\sc $k$-Cycle} problems which are NP-complete~\cite{fptbook}. Here, the input is a (di)graph $G$ and an integer $k$, and the objective is to determine of there is a path (cycle) on at least $k$-distinct vertices of the graph. These problems can be solved in $2^{\mathcal{O}(k)}\cdot {\sf \textup{poly}(n)}$ time using techniques based on color-coding, hash-functions and splitters~\cite{fptbook,alonColorCoding,naorSplitters}. In particular, we can reduce {\sc $k$-Path} to the following problem in polynomial time: find a $k$-path in a \emph{colorful} graph on $n$ vertices, whose vertices have been colored with $\mathcal{O}({\sf \textup{poly}}(k) \cdot \log n)$ colors, such that every vertex of the $k$-path has a distinct color. However, for our purposes the construction of the cycle $C$ in $G$ is a part of the proof of Theorem~\ref{mainthm3body} (described in Section~\ref{boundonR}: we show that in polynomial time, one can find a cycle in a $(d^4+d)$-partite digraph, in which each part has at most $d $ vertices and for any two parts $V$ and $V'$ in the digraph, every vertex in $V'$ has an incoming edge from some vertex in $V$ and vice-versa. This implies that if $\abs{L_X} > d^4 +d$, then in polynomial time, we can determine a cycle $C$ in $G$ that visits each part at most once and then determine a $(1-\varepsilon)$-EFX allocation $X' \pd X$ by applying $U_3$. This also implies that $\ram(d) \leq d^4 + d$. Therefore, 

\begin{lemma}
	\label{L_Xboundalg}
	Consider a $(1-\varepsilon)$-EFX allocation $X$. If $\abs{L_X} >  d^4 + d$, then in polynomial time, we can determine a $(1-\varepsilon)$-EFX allocation $X'\pd X$.
\end{lemma}

\paragraph{Putting it together.} We give the existence proof and indicate in brackets the changes required for the polynomial time algorithm. We start with an empty allocation, which is trivially a $(1-\varepsilon)$-EFX. Then, our algorithm iteratively maintains a $(1-\varepsilon)$-EFX allocation $X$ and a pool of unallocated goods. In each iteration, the algorithm first makes $E_X$ acyclic in polynomial time (Lemma~\ref{envycycleelimination}). Thereafter,  our algorithm checks whether any one of the update rules $U_1$ and $U_2$ is applicable. If $U_1$ is applicable, then our algorithm determines an allocation a $(1-\varepsilon)$-EFX allocation $X'$ where $v_i(X'_i) \geq v_i(X_i)$ for all $i \in [n]$ and  the number of unallocated goods reduces. If $U_2$ is applicable, then our algorithm determines a $(1-\varepsilon)$-EFX allocation $X' \pd X$.  If neither $U_1$ nor $U_2$ is applicable, then it determines the sets $H_X$ and $L_X$. By Lemma~\ref{V_Xbound}, we have $\abs{H_X} \leq 2n/(\varepsilon \cdot d)$. If $\abs{L_X} \leq \ram(d)$ ($\abs{L_X} \leq d^4 + d$), then it returns the allocation $X$. Otherwise  it determines a cycle that visits each part of $G$ at most once and then determines $(1-\varepsilon)$-EFX allocation $X' \pd X$ by applying update rule $U_3$ (by Lemma~\ref{L_Xbound}). If $\abs{L_X} > d^4 + d$, the cycle can be determined in polynomial time. 
Therefore, when the algorithm terminates, we have that $\abs{H_X} \leq 2n/(\varepsilon \cdot d)$ and $\abs{L_X} \leq \ram(d)$, ($\abs{L_X} \le d^4 + d$)  implying that the total number of unallocated goods is $\abs{H_X} + \abs{L_X} \leq 2 \cdot \max( 2n/ (\varepsilon \cdot d), \ram(d))$ ($2 \cdot \max(2n/(\eps \cdot d), 2d^4)$).

We now state the explicit value of $d$, first for the existence proof. We choose $d$ as the smallest integer such that  $2n / (\varepsilon d) \le  \ram(d)$, i.e, $d = h^{-1} (2n / \varepsilon)$. \footnote{Recall that $h(d) = d \cdot \ram(d)$ in Definition~\ref{Rdef} and that $h^{-1} (2n/\eps)$ is defined as the smallest integer such that $h(d) \ge 2n/\eps$. } Therefore, the number of unallocated goods is at most $4 n /(\eps \cdot h^{-1}(2n / \eps)) $.

For the algorithmic result, we choose $d$ as the smallest integer such that $2n/(\eps \cdot d) \le 2d^4$. Then $d =
\ceil{(n/\eps)^{1/5}}$ and the number of unallocated goods is at most $4 \ceil{(n/\eps)^{1/5}}^4$. This is less than $64 (n/\eps)^{4/5}$.

It only remains to show that the algorithm will terminate. We prove a polynomial bound on the number of iterations. The bound applies to the existence and the algorithmic version. To this end, note that in each iteration, after removing cycles from $E_X$, our algorithm determines a new $(1-\varepsilon)$-EFX allocation $X'$ through one of the following procedures:
\begin{itemize}
	\item applying $U_1$,
	\item applying $U_2$, 
	\item determining a cycle $C$ that visits each part in $G$ at most once and then applying $U_3$. 
\end{itemize}
Note that the initial envy-cycle elimination and subsequent application of all of the above procedures ensure that $v_i(X'_i) \geq v_i(X_i)$ for all $i \in [n]$ (Lemmas~\ref{envycycleelimination},~\ref{U1},~\ref{U2},~\ref{U3}). Thus, throughout the algorithm the valuation of an agent never decreases. Note that there cannot be more than $m$ \emph{consecutive applications} of $U_1$, as the number of unallocated goods decreases with each application of $U_1$. Every time we apply $U_2$ or $U_3$, we ensure that $X' \pd X$, implying that the valuation of some agent improves by a factor of at least $(1+\varepsilon)$. Since each agent's valuation is bounded by $W = \mathit{max}_{i \in [n]} v_i(M)$, and the valuation of an agent never decreases throughout the algorithm, we can have at most $\textup{poly}(n,m,W, 1 / \varepsilon)$ many iterations that involve applications of $U_2$ and $U_3$.  Therefore, the total number of iterations of our algorithm is $m \cdot (\text{iterations involving application of $U_2$ or $U_3$})$ which is also $\textup{poly}(n,m,\log W, 1 / \varepsilon)$. Notice that in the algorithmic case, each of the iterations can also be implemented in polynomial time: $U_1$ and $U_2$ can be implemented in polynomial time (Lemmas~\ref{U1} and~\ref{U2}). When $\abs{L_X} \geq 2d^4 \geq d^4 + d$, then in polynomial time we can determine the cycle $C$ and apply $U_3$ (Lemma~\ref{L_Xboundalg}).  We can now state the main result of this section.

\begin{theorem}
	\label{mainthm2version2}
	Let $h(d) = d \cdot \ram(d)$. Then there is a $(1- \varepsilon)$-EFX allocation $X$ and a set of unallocated goods $P$ such that $\abs{P} \leq ({4n} / {(\eps \cdot h^{-1}(n / (\varepsilon)))}$. In polynomial time, one can find a $(1 - \eps)$-EFX allocation and a set $P$ of unallocated goods such that $\abs{P} \le 64(n/\eps)^{4/5}$. 
\end{theorem}

Note that any upper bound on the rainbow cycle number will imply an upper bound on the number of unallocated goods. %In Section~\ref{boundonR} we show that $\ram(d) \in \mathcal{O}(d^4)$.   

\section{Bounds on the Rainbow Cycle Number}
\label{boundonR}
In this section, we give the proof of Theorem~\ref{mainthm3}. We briefly recall the setup: There is a $k$-partite digraph $G = (\cup_{i \in [k]} V_i, E_G)$ such that each part has at most $d$ vertices. For every distinct parts  $V_i$ and $V_j$, every vertex in $V_j$ has an incoming edge from some vertex in $V_i$. There is no cycle in $G$ that visits each part at most once. Our goal is to establish an upper bound on $k$.

We now introduce some helpful notations and concepts. For each $i \in [k]$, we represent the vertices in the part $V_i$ as $(i, \textup{vertex id})$, i.e,  $V_i = \{ (i,1), (i,2), \dots , (i, \lvert V_i \rvert) \}$. For any positive integer $d$ and  $a,b \in [d]$, we use  $\sigma_d(a,b)$ to denote $(a-1) \cdot d + b$.  Note that $1 \le \sigma_d(a,b) \le d^2$. The $\sigma_d(a,b)$ captures the lexicographic ordering among the pairs $\cup_{a \in [d]} \cup_{b \in [d]} (a,b)$. For any Boolean vector $u \in \{0,1\}^r$, we use $u[k]$ to refer to the $k^{\mathit{th}}$ coordinate of the vector $u$. We introduce the simple yet crucial notion of \emph{representative set for a set of Boolean vectors}. Given a set $D$ of $r$-dimensional Boolean vectors, the set $B \subseteq D$ is a \emph{representative set} of $D$,
if $\set{\ell}{\text{$a[\ell] = 1$ for some $a \in D$}}=  \set{\ell}{\text{$b[\ell] = 1$ for some $b \in B$}}$. We first make an observation about the size of $B$.

\begin{observation}
	\label{sizeofB}
	Given any set $D$ of $r$-dimensional Boolean vectors, there exists a representative set $B \subseteq D$ of size at most $r$.
\end{observation}

\begin{proof}
	For each coordinate $\ell \in [r]$ we do: if there is a vector $a \in D$ with $a[\ell] = 1$, we put one such vector into $B$. Clearly, $\abs{B} \le r$.
\end{proof}

We prove Theorem~\ref{mainthm3} by contradiction. To be precise, we show that if $k > d^4 +d $, then there exists a cycle in $G$ that visits every part at most once. Moreover, this cycle can be found in time polynomial in $k$.

We construct the cycle in two steps. We first show the existence of a part $V_{\tilde{\ell}}$ such that there is a directed cycle that visits only the parts $V_{\tilde{\ell}}$, $V_1$, $V_2$, \ldots, $V_d$ and moreover each of the parts $V_1$, $V_2$, \ldots, $V_d$ at most once. In the second step we replace the vertices in $V_{\tilde{\ell}}$
in this cycle by vertices in distinct parts. 

For each ordered pair $(i,j) \in [d] \times [d]$, and  $\ell \in [k] \setminus [d]$, we define a $d^2$-dimensional vector $u_{i,j,\ell}$ as follows: for all $x \in [d]$ and $y \in [d]$, we set $u_{i,j,\ell}[\sigma_d(x,y)] = 1$ if and only if there exists a path $(i,x) \rightarrow (\ell,z) \rightarrow (j,y)$ in $G$ for some $(\ell,z) \in V_{\ell}$, i.e., if there exists a path from vertex $(i,x)$ in  $V_i$ to vertex $(j,y)$ in $V_j$ through some vertex in $V_{\ell}$. Otherwise, we set $u_{i,j,\ell}[\sigma_d(x,y)] = 0$.

\renewcommand{\L}{\mathcal{L}}

Let $\L = [k] \setminus [d]$. For each ordered pair $(i,j) \in [d] \times [d]$, we construct the sets $B^{i,j}$ and $\L^{i,j}$ as follows: For each $(i,j)$ taken in the increasing order of $\sigma_{d}(i,j)$, define $\L^{i,j} = \L$ and $B^{i,j}$ as a representative vector set of $\set{u_{i,j,\ell}}{\ell \in \L^{i,j}}$ of size at most $d^2$. A set $B^{i,j}$ of this size exists because our vectors have dimension $d^2$. Then we set $\L = \L \setminus \set{\ell}{u_{i,j, \ell} \in B^{i,j}}$. At most $d^2$ elements are removed from $\L$ in each iteration.

For clarity, we write $\L^{f}$ to denote the set $\L$ at the end of the construction. Observe that $\abs{\L^f} \geq 1$. This holds since we start with a set of size larger than $d^4$ and removed at most $d^2$ elements in each of the $d^2$ iterations.

\begin{observation}
	\label{disjointbypasses}
	Consider distinct ordered pairs $(i,j) \in [d] \times [d]$ and $(i',j') \in [d] \times [d]$.  The sets $\set{\ell}{u_{i,j,\ell} \in B^{i,j}}$ and $\set{\ell}{u_{i',j',\ell} \in B^{i',j'}}$ are disjoint.  
\end{observation}

\begin{proof}
	Let us assume without loss of generality that $\sigma_{d}(i,j) < \sigma_{d}(i',j')$. Consider any $\ell$ such that $u_{i,j,\ell} \in B^{i,j}$. Then $\ell$ is removed from $\L$ at the end of the iteration for the pair $(i,j)$ and hence does not belong to $\L$ at the beginning of the iteration for the pair $(i',j')$. Consequently $u_{i',j', \ell} \notin B^{i',j'}$ (by definition of $B^{i',j'}$, if $u_{i',j',\ell} \in B^{i',j'}$, then $\ell \in \L^{i',j'}$).
\end{proof}

At the end of the construction, we arbitrarily pick a $\tilde{\ell} \in \L^f$ (this is possible as $\L^f \neq \emptyset$). Now, we make a small observation about the vector $u_{i,j,\tilde{\ell}}$ for all $i,j \in [d]$. 

\begin{observation}
	\label{basis}
	For all $i,j \in [d]$, if $u_{i,j, \tilde{\ell}}[q] = 1$ for some $q \in [d^2]$, then there exists a vector $u_{i,j,{l'}} \in B^{i,j}$ such that $u_{i,j,{l'}}[q] = 1$.
\end{observation} 

\begin{proof}
	Observe that $\L^f \subseteq \L^{i,j}$. Therefore, $\tilde{l} \in \L^{i,j}$. By definition, $B^{i,j}$ is a representative vector set of  $\set{u_{i,j, \ell}}{\ell \in \L^{i,j}}$.  Therefore, by the definition of representative set, there exists a vector $u_{i,j,{\ell'}} \in B^{i,j}$ such that $u_{i,j,{\ell'}}[q]=1$. 
\end{proof}

We are now ready for the construction of a cycle that visits each part at most once. We first show that there exists a cycle $C$ in $G$ that visits only the parts $V_{\tilde{\ell}}$, $V_1$, \ldots, $V_d$ and each of the parts $V_1$, \ldots $V_d$ at most once, i.e, the only part it may visit more than once is $V_{\tilde{\ell}}$. See Figure~\ref{complexcycle} for an illustration.

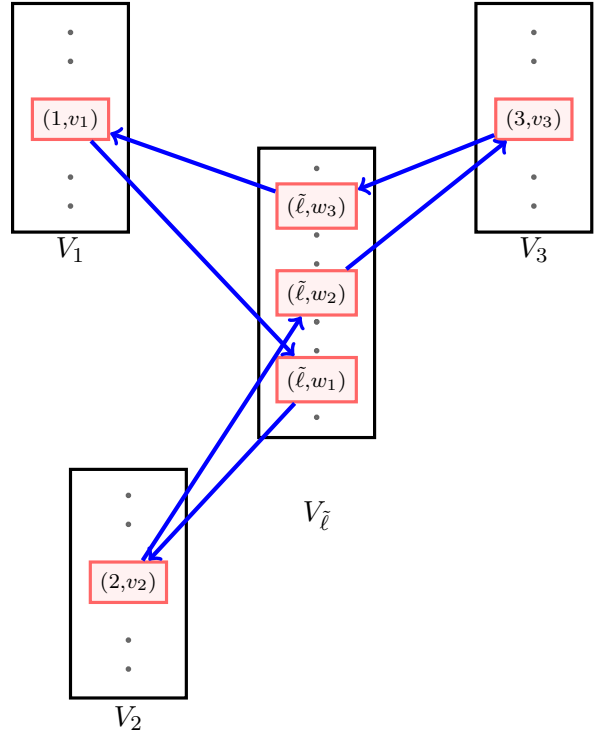
\begin{figure}[tb]
	\centering
	\begin{tikzpicture}[scale=0.77,
roundnode/.style={circle, draw=green!60, fill=green!5, very thick, minimum size=7mm},
squarednode/.style={rectangle, draw=red!60, fill=red!5, very thick, minimum size=3mm},
]

%Vertices
\node[squarednode]      (v1)      at (4.75,8.5)                       {$\scriptstyle{(1,v_1)}$};
\node[squarednode]      (v2)      at (5.75,0.5)                       {$\scriptstyle{(2,v_2)}$};
\node[squarednode]      (v3)      at (12.75,8.5)                       {$\scriptstyle{(3,v_3)}$};

\node[squarednode]      (w1)      at (9,4)                       {$\scriptstyle{(\tilde{\ell},w_1)}$};
\node[squarednode]      (w2)      at (9,5.5)                       {$\scriptstyle{(\tilde{\ell},w_2)}$};
\node[squarednode]      (w3)      at (9,7)                       {$\scriptstyle{(\tilde{\ell},w_3)}$};

%Partites
\draw[black, very thick] (3.75,4.55+2) rectangle (5.75,8.5+2);
\node at (4.75,6.25) {$V_1$};
\draw[black, very thick] (4.75,0.5-2) rectangle (6.75,4.45-2);
\node at (5.75,-1.85) {$V_2$};
\draw[black, very thick] (11.75,4.55+2) rectangle (13.75,8.5+2);
\node at (12.75,6.25) {$V_3$};

\draw[black, very thick] (8,8) rectangle (10,3);
\node at (9,1.65) {$V_{\tilde{\ell}}$};

%Dots
\filldraw[color=black!60, fill=black!5, very thick](4.75,10) circle (0.02);
\filldraw[color=black!60, fill=black!5, very thick](4.75,9.5) circle (0.02);
\filldraw[color=black!60, fill=black!5, very thick](4.75,7.5) circle (0.02);
\filldraw[color=black!60, fill=black!5, very thick](4.75,7) circle (0.02);

\filldraw[color=black!60, fill=black!5, very thick](5.75,2) circle (0.02);
\filldraw[color=black!60, fill=black!5, very thick](5.75,1.5) circle (0.02);
\filldraw[color=black!60, fill=black!5, very thick](5.75,-0.5) circle (0.02);
\filldraw[color=black!60, fill=black!5, very thick](5.75,-1) circle (0.02);

\filldraw[color=black!60, fill=black!5, very thick](12.75,10) circle (0.02);
\filldraw[color=black!60, fill=black!5, very thick](12.75,9.5) circle (0.02);
\filldraw[color=black!60, fill=black!5, very thick](12.75,7.5) circle (0.02);
\filldraw[color=black!60, fill=black!5, very thick](12.75,7) circle (0.02);

\filldraw[color=black!60, fill=black!5, very thick](9,3.35) circle (0.02);
\filldraw[color=black!60, fill=black!5, very thick](9,5) circle (0.02);
\filldraw[color=black!60, fill=black!5, very thick](9,4.5) circle (0.02);
\filldraw[color=black!60, fill=black!5, very thick](9,5) circle (0.02);
\filldraw[color=black!60, fill=black!5, very thick](9,6) circle (0.02);
\filldraw[color=black!60, fill=black!5, very thick](9,6.5) circle (0.02);
\filldraw[color=black!60, fill=black!5, very thick](9,7.65) circle (0.02);

%Complexcycle
\draw[<-, ultra thick, blue] (w3) -- (v3);
\draw[<-, ultra thick, blue] (v3) -- (w2);
\draw[<-, ultra thick, blue] (w2) -- (v2);
\draw[<-, ultra thick, blue] (v2) -- (w1);
\draw[<-, ultra thick, blue] (w1) -- (v1);
\draw[<-, ultra thick, blue] (v1) -- (w3);
\end{tikzpicture}
	\caption{Illustration of the first part of the construction. The cycle in the figure visits the parts  $V_1$, $V_2$ and $V_3$ exactly once and the part $V_{\tilde{\ell}}$ three times. It is given by $(\tilde{\ell},w_3) \rightarrow (1,v_1) \rightarrow (\tilde{\ell},w_1) \rightarrow (2,v_2) \rightarrow (\tilde{{\ell}}, w_2) \rightarrow (3,v_3) \rightarrow (\tilde{\ell},w_3)$.  }
	\label{complexcycle}
\end{figure}

Let $({\tilde{\ell}}, w_{d})$ be an arbitrary vertex in $V_{\tilde{\ell}}$. We construct a path 
\[ (\tilde{\ell},w_0) \rightarrow (1,v_1)  \rightarrow   \ldots \rightarrow (i-1,v_{i-1}) \rightarrow (\tilde{\ell},w_{i-1}) \rightarrow (i,v_{i}) \rightarrow (\tilde{\ell},w_{i}) \rightarrow \ldots \rightarrow (d,v_{d}) \rightarrow (\tilde{\ell},w_{d}) \]
by starting at $(\tilde{\ell},w_{d})$ and tracing backwards: We start in $(\tilde{\ell},w_{d})$. Assume that we already traced back to $(\tilde{\ell},w_{i})$ with $i = d$ initially. By the construction of $G$, there must be an edge from some vertex $(i,v_{i})$ in $V_{i}$ to $(\tilde{\ell},w_{i})$ in $V_{\tilde{\ell}}$, and there must be an edge from some vertex $(\tilde{\ell},w_{i-1})$ in  $V_{\tilde{\ell}}$ to $(i,v_i)$ in $V_i$. Thus there is the path $(\tilde{\ell},w_{i-1}) \rightarrow (i,v_{i}) \rightarrow (\tilde{\ell},w_{i})$ in $G$. We keep continuing this procedure until we reach $(\tilde{\ell},w_0)$.

Since the part $V_{\tilde{\ell}}$ can have at most $d$ vertices, by the pigeonhole principle, there must be $i$ and $j$ with $0 \le i < j \le d$ such that $w_i = w_j$. Let $C$ be the subpath from $(\tilde{\ell},w_i)$ to $(\tilde{\ell},w_j)$, i.e.,
\[   C = (\tilde{\ell},w_i) \rightarrow (i+1,v_{i+1}) \rightarrow (\tilde{\ell},w_{i+1}) \rightarrow \ldots \rightarrow (\tilde{\ell},w_{j-1}) \rightarrow (j,v_j) \rightarrow (\tilde{\ell},w_j).\] 

Observe that $C$ visits all the parts  of $G$ except $V_{\tilde{\ell}}$ at most once.  We now show that by using ``bypass" parts  we can make the cycle simple. For clarity, we rewrite $C$ as
\[  C = (i+1,v_{i+1}) \rightarrow (\tilde{\ell},w_{i+1}) \rightarrow \ldots \rightarrow (\tilde{\ell},w_{j-1}) \rightarrow (j,v_j) \rightarrow (\tilde{\ell},w_j) \rightarrow (i+1,v_{i+1}).\]

\begin{figure}[tb]
	\centering 	
	\begin{tikzpicture}[scale=0.9,
roundnode/.style={circle, draw=green!60, fill=green!5, very thick, minimum size=7mm},
squarednode/.style={rectangle, draw=red!60, fill=red!5, very thick, minimum size=5mm},
]

%Vertices
\node[squarednode]      (v1)      at (4.75,8.5)                       {$\scriptstyle{(1,v_1)}$};
\node[squarednode]      (v2)      at (5.75,0.5)                       {$\scriptstyle{(2,v_2)}$};
\node[squarednode]      (v3)      at (12.75,8.5)                       {$\scriptstyle{(3,v_3)}$};

\node[squarednode]      (w1)      at (9,4)                       {$\scriptstyle{(\tilde{\ell},w_1)}$};
\node[squarednode]      (w2)      at (9,5.5)                       {$\scriptstyle{(\tilde{\ell},w_2)}$};
\node[squarednode]      (w3)      at (9,7)                       {$\scriptstyle{(\tilde{\ell},w_3)}$};

\node[squarednode]      (l1)      at (1.95,5.25)                       {$\scriptstyle{(\ell_1,y_1)}$};
\node[squarednode]      (l2)      at (12.05,1.25)                       {$\scriptstyle{(\ell_2,y_2)}$};
\node[squarednode]      (l3)      at (9,11.75)                       {$\scriptstyle{(\ell_3,y_3)}$};

%Partites
\draw[black, very thick] (3.75,4.55+2) rectangle (5.75,8.5+2);
\node at (4.75,6.25) {$V_1$};
\draw[black, very thick] (4.75,0.5-2) rectangle (6.75,4.45-2);
\node at (5.75,-1.85) {$V_2$};
\draw[black, very thick] (11.75,4.55+2) rectangle (13.75,8.5+2);
\node at (12.75,6.25) {$V_3$};
\draw[black, very thick] (8,8) rectangle (10,3);
\node at (9,1.65) {$V_{\tilde{\ell}}$};

\draw[black, very thick] (1,7) rectangle (3,3.5);
\node at (1.5,3.25) {$V_{\ell_1}$};
\draw[black, very thick] (11,3) rectangle (13,-0.5);
\node at (12,-0.75) {$V_{\ell_2}$};
\draw[black, very thick] (8,10) rectangle (10,13.5);
\node at (9,9.75) {$V_{\ell_3}$};

%Dots
\filldraw[color=black!60, fill=black!5, very thick](4.75,10) circle (0.02);
\filldraw[color=black!60, fill=black!5, very thick](4.75,9.5) circle (0.02);
\filldraw[color=black!60, fill=black!5, very thick](4.75,7.5) circle (0.02);
\filldraw[color=black!60, fill=black!5, very thick](4.75,7) circle (0.02);

\filldraw[color=black!60, fill=black!5, very thick](5.75,2) circle (0.02);
\filldraw[color=black!60, fill=black!5, very thick](5.75,1.5) circle (0.02);
\filldraw[color=black!60, fill=black!5, very thick](5.75,-0.5) circle (0.02);
\filldraw[color=black!60, fill=black!5, very thick](5.75,-1) circle (0.02);

\filldraw[color=black!60, fill=black!5, very thick](12.75,10) circle (0.02);
\filldraw[color=black!60, fill=black!5, very thick](12.75,9.5) circle (0.02);
\filldraw[color=black!60, fill=black!5, very thick](12.75,7.5) circle (0.02);
\filldraw[color=black!60, fill=black!5, very thick](12.75,7) circle (0.02);

\filldraw[color=black!60, fill=black!5, very thick](9,3.5) circle (0.02);
\filldraw[color=black!60, fill=black!5, very thick](9,5) circle (0.02);
\filldraw[color=black!60, fill=black!5, very thick](9,4.5) circle (0.02);
\filldraw[color=black!60, fill=black!5, very thick](9,5) circle (0.02);
\filldraw[color=black!60, fill=black!5, very thick](9,6) circle (0.02);
\filldraw[color=black!60, fill=black!5, very thick](9,6.5) circle (0.02);
\filldraw[color=black!60, fill=black!5, very thick](9,7.5) circle (0.02);

\filldraw[color=black!60, fill=black!5, very thick](1.95,6.5) circle (0.02);
\filldraw[color=black!60, fill=black!5, very thick](1.95,5.75) circle (0.02);
\filldraw[color=black!60, fill=black!5, very thick](1.95,4.5) circle (0.02);
\filldraw[color=black!60, fill=black!5, very thick](1.95,3.75) circle (0.02);

\filldraw[color=black!60, fill=black!5, very thick](12.05,2.75) circle (0.02);
\filldraw[color=black!60, fill=black!5, very thick](12.05,2) circle (0.02);
\filldraw[color=black!60, fill=black!5, very thick](12.05,0.5) circle (0.02);
\filldraw[color=black!60, fill=black!5, very thick](12.05,-0.25) circle (0.02);

\filldraw[color=black!60, fill=black!5, very thick](9,13.25) circle (0.02);
\filldraw[color=black!60, fill=black!5, very thick](9,12.5) circle (0.02);
\filldraw[color=black!60, fill=black!5, very thick](9,11) circle (0.02);
\filldraw[color=black!60, fill=black!5, very thick](9,10.25) circle (0.02);

%Complexcycle
\draw[<-, ultra thick, lightgray] (w3) -- (v3);
\draw[<-, ultra thick, lightgray] (v3) -- (w2);
\draw[<-, ultra thick, lightgray] (w2) -- (v2);
\draw[<-, ultra thick, lightgray] (v2) -- (w1);
\draw[<-, ultra thick, lightgray] (w1) -- (v1);
\draw[<-, ultra thick, lightgray] (v1) -- (w3);

%Simple cycle
\draw[->, ultra thick, blue] (v1)--(l1);
\draw[->, ultra thick, blue] (l1)--(v2);
\draw[->, ultra thick, blue] (v2)--(l2);
\draw[->, ultra thick, blue] (l2)--(v3);
\draw[->, ultra thick, blue] (v3)--(l3);
\draw[->, ultra thick, blue] (l3)--(v1);

\end{tikzpicture}
	\caption{Illustration of the existence of a cycle that visits every part at most once. We take the instance in Figure~\ref{complexcycle}, where there exists a cycle $C$ that visits every part other than $V_{\tilde{\ell}}$ at most once. The edges of the cycle $C$ are light gray color in color. The figure shows how to obtain a cycle $C'$ that visits every part at most once from $C$. The edges of $C'$ are blue in color. For all $i \in [3]$, we replace the subpath in $C$ of the form $(i,v_i) \rightarrow (\tilde{\ell},w_i) \rightarrow (i+1,v_{i+1})$ ($3+1$ is to be interpreted as $1$) by $(i,v_i) \rightarrow (\ell_i,y_i) \rightarrow (i+1, v_{i+1})$ to get $C'$.} 
	\label{simplecycle}
\end{figure}

\paragraph{Making the Cycle Simple.} For all $q \in [i+1,j]$ consider the subpath
\[  (q,v_q) \rightarrow (\tilde{\ell},w_{q}) \rightarrow (q + 1, v_{q + 1}) \]
of $C$ (index $j+1$ is to be interpreted as $i+1$). The existence of such a subpath in $G$ implies that $u_{q,q+1,\tilde{\ell}}[\sigma_d(v_q,v_{q+1})] = 1$. By Observation~\ref{basis}, we know that there is a vector  $u_{q, q+1,\ell_q} \in B^{q,q+1}$ such that $u_{q, q+1, \ell_q}[\sigma_d(v_q,v_{q+1})] = 1$. This implies that there exists a part $V_{\ell_q}$, and a vertex $(\ell_q, y_{q})$ in part  $V_{\ell_q}$, such that there is a subpath 
\[  (q,v_q) \rightarrow ({\ell_q},y_{q}) \rightarrow (q + 1, v_{q + 1})\enspace . \]

By Observation~\ref{disjointbypasses}, we have that $\ell_q \neq \ell_{q'}$ for all $q \neq q'$. Therefore we have a simple cycle $C'$ in $G$ that visits each part in $G$ at most once, namely,
\[C' = (i+1,v_{i+1}) \rightarrow (\ell_{i+1},y_{i+1}) \rightarrow \dots \rightarrow (\ell_{j-1},y_{j-1}) \rightarrow (j,v_j) \rightarrow (\ell_j,y_j) \rightarrow (i+1,v_{i+1}). \]
See Figure~\ref{simplecycle} for an illustration of this entire procedure.

Therefore if $k > d^4 + d$, then there exists a cycle in $G$ that visits each part at most once. Moreover, this cycle can be found in time polynomial in $k$. With this we arrive at the main result of this section. 

\begin{theorem}
	\label{mainthm3body}
	For all $d \ge 1$, we have $\ram(d) \leq d^4 + d$. Furthermore, Let $G$ be a $k$-partite digraph with $k > d^4 + d$ parts of cardinality at most $d$ each, such that for every vertex $v$ and any part $W$ not containing $v$, there is an edge from $W$ to $v$. Then, there exists a cycle in $G$ visiting each part at most once, and it can be found in time polynomial in $k$.
\end{theorem}

An improved upper bound on $\ram(d)$ would imply a better bound on the number of unallocated goods. However, we show that an exponential improvement (e.g. $\ram(d) \in \textup{poly}(\log(d))$) is not possible by showing a \emph{linear} lower bound, i.e., $\ram(d) \geq d$. However, this still leaves room for polynomial improvement and we suspect that $\ram(d) \in \mathcal{O}(d)$. This would imply the existence of a $(1-\varepsilon)$-EFX allocation with $\mathcal{O}(\sqrt{n/\varepsilon})$ many goods unallocated.  For a polynomial time algorithm, the construction of a cycle as in Theorem~\ref{mainthm3body} would have to be polynomial time. However, we remark that this is an initiation study for determining $(1-\varepsilon)$-EFX allocations with sublinear number of unallocated goods and we use concepts like the \emph{group champion graph} that are natural extensions of the champion graph. We believe that this still leaves room for developing more sophisticated concepts and techniques that may reduce the number of unallocated goods to $o(\sqrt{n/\varepsilon})$.

\paragraph{Lower bound on $\ram(d)$.} We show that $\ram(d) \geq d$. We construct a $d$-partite graph $G =(\cup_{i \in [d]} V_i, E)$ such that each part $V_i$ has $d$ vertices, for all pairs of parts $V_i$ and $V_j$, every vertex in $V_j$ has an incoming edge from a vertex in $V_i$ and vice-versa, and there exists no cycle that visits each part at most once.

We now define the edges in $G$. Let $V_i = \{(i,0), (i,1), \dots, (i,d-1) \}$. Consider any $i$ and $j$ such that $i<j$ . For each $0 \leq \ell \leq d-1$, we have an edge from $(i,\ell)$ in $V_i$ to $(j,\ell)$ in $V_j$ and there is an edge from $(j,\ell)$ in $V_{j}$ to $(i,(\ell + 1) \bmod d)$ in $V_i$ (see Figure~\ref{lower-bound} for an illustration). One can easily verify that for all parts $V_i$ and $V_j$, every vertex in part $V_j$ has an incoming edge from part $V_i$ and vice-versa. It suffices to show that $G$ admits no cycle that visits each part at most once.

\begin{figure}[H]
	\centering
	\begin{tikzpicture}
	[roundnode/.style={circle, draw=green!60, fill=green!5, very thick, minimum size=7mm},
	agent/.style={rectangle, draw=red!60, fill=red!5, very thick, minimum size=5mm},
	]
	% % % % % % % % % % % % % % % % % %D=2		
	%Parts
	\draw[black, very thick] (-0.5,0.5) rectangle (0.5+0.05,-2.5);
	\draw[black, very thick] (-0.5+2,0.5) rectangle (0.5+2+0.05,-2.5);
	\node at (0,-2.75) {$V_1$};
	\node at (2,-2.75) {$V_2$};
	
	%Vertices for config1
	\node[agent]      (a1) at (0,0)      {$\scriptstyle{(1,0)}$};
	\node[agent]      (a2) at (0,-2)      {$\scriptstyle{(1,1)}$};
	\node[agent]      (b1) at (2,0)     {$\scriptstyle{(2,0)}$};
	\node[agent]      (b2) at (2,-2)     {$\scriptstyle{(2,1)}$};

	%Edges
	\draw[->,blue,thick] (a1) -- (b1);
	\draw[->,blue,thick] (b2) -- (a1);
	\draw[->,blue,thick] (a2) -- (b2);
	\draw[->,blue,thick] (b1) -- (a2);

	% % % % % % % % % % % % % % % % % %D=3
	%Parts
	\draw[black, very thick] (6-0.5,0.5) rectangle (6+0.5+0.05,-4.5);
	\draw[black, very thick] (-0.5+8,0.5) rectangle (0.5+8+0.05,-4.5);
	\draw[black, very thick] (-0.5+10,0.5) rectangle (0.5+10+0.05,-4.5);
	\node at (6,-4.75) {$V_1$};
	\node at (8,-4.75) {$V_2$};
	\node at (10,-4.75) {$V_3$};
	
	%Vertices for config1
	\node[agent]      (a1) at (6,0)      {$\scriptstyle{(1,0)}$};
	\node[agent]      (a2) at (6,-2)      {$\scriptstyle{(1,1)}$};
	\node[agent]      (a3) at (6,-4)      {$\scriptstyle{(1,2)}$};

	\node[agent]      (b1) at (8,0)     {$\scriptstyle{(2,0)}$};
	\node[agent]      (b2) at (8,-2)     {$\scriptstyle{(2,1)}$};
	\node[agent]      (b3) at (8,-4)      {$\scriptstyle{(2,2)}$};

	\node[agent]      (c1) at (10,0)     {$\scriptstyle{(3,0)}$};
	\node[agent]      (c2) at (10,-2)     {$\scriptstyle{(3,1)}$};
	\node[agent]      (c3) at (10,-4)      {$\scriptstyle{(3,2)}$};

	%Edges
	\draw[->,blue,thick] (a1) -- (b1);
	\draw[->,blue,thick] (a2) -- (b2);
	\draw[->,blue,thick] (a3) -- (b3);
	
	\draw[->,blue,thick] (b1) -- (c1);
	\draw[->,blue,thick] (b2) -- (c2);
	\draw[->,blue,thick] (b3) -- (c3);
	
	\draw[->,blue,thick] (b1) -- (a2);
	\draw[->,blue,thick] (b2) -- (a3);
	\draw[->,blue,thick] (b3) -- (a1);
	
	\draw[->,blue,thick] (c1) -- (b2);
	\draw[->,blue,thick] (c2) -- (b3);
	\draw[->,blue,thick] (c3) -- (b1);
	
	\path[blue, thick, ->, out=20,in=160]    (a1) edge (c1);
	\path[blue, thick, ->, out=20,in=160]    (a2) edge (c2);
	\path[blue, thick, ->, out=20,in=160]    (a3) edge (c3);
	
	\draw[->,blue,thick] (c1) -- (a2);
	\draw[->,blue,thick] (c2) -- (a3);
	\path[blue, thick, ->, out=125,in=-25]    (c3) edge (a1);
	\end{tikzpicture}
	\caption{Illustration of the construction of $d$-partite graph $G$ that satisfies all the conditions in Definition~\ref{Rdef}, for $d=2$ (left) and $d=3$ (right). } 
	\label{lower-bound}	
\end{figure}
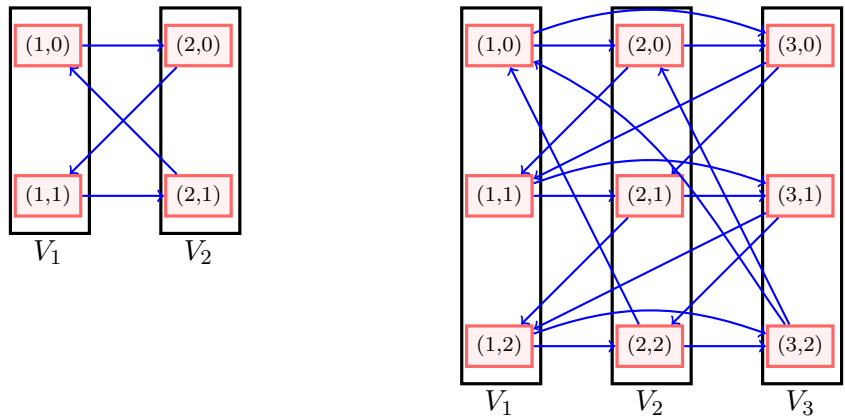

\begin{lemma}
	There exists no cycle in $G$ that visits each part at most once.
\end{lemma}

\begin{proof}
	We prove by contradiction. Assume that there is a cycle $C = (i_1, \ell_1) \rightarrow (i_2, \ell_2) \rightarrow \dots \rightarrow (i_r,\ell_r) \rightarrow (i_1,\ell_1)$ that visits each part at most once, i.e., $i_1\neq i_2 \neq \dots \neq i_r$. From here on, all the indices are modulo $r$. Note that by the construction of the edges of $G$, for all $q \in [r]$, we have $\ell_{q+1} = \ell_q$ if $i_q < i_{q+1}$ and  $\ell_{q+1}  = (\ell_{q} + 1)\bmod d$ if $i_q > i_{q+1}$. Let $\#_1 = \{q \in [r] \mid i_q > i_{q+1} \}$ (recall that $r+1$ is $1$). The existence of the cycle $C$ in $G$ implies that $\ell_1 = (\ell_1 + \#_1) \bmod d $.
	
	Since $i_1\neq i_2 \neq \dots \neq i_r$ and there exists the cycle $C$ in $G$, there are  indices $q'$ and $q''$ such that $i_{q'} > i_{q'+1}$ and $i_{q''} < i_{q''+1}$, further implying that $1 \leq \#_1 \leq r-1$.  Since $G$ has $d$ parts, we have $r \leq d$, implying that  $1 \leq \#_1 \leq d-1$. However this implies that $(\ell_1 + \#_1) \mod d  \neq \ell_1$, which is a contradiction.   
\end{proof}

\section{Finding \emph{Efficient} $(1-\varepsilon)$-EFX Allocations with Sublinear Number of Unallocated Goods} 
\label{efficiency}
We note that like the algorithms in ~\cite{CKMS20, TimPlaut18}, our algorithm is flexible with the initialization, i.e., starting with any initial $(1-\varepsilon)$-EFX allocation $X$, it can determine a final $(1-\varepsilon)$-EFX allocation $Y$ with at most $\mathcal{O}((n / \varepsilon)^{\tfrac{4}{5}})$ many goods unallocated and $v_i(Y_i) \geq v_i(X_i)$ for all $i \in [n]$. This is consequence of the fact that the valuation of an agent never decreases throughout our algorithm. Therefore, our algorithm maintains the welfare of the initial allocation. Thus, if we choose the initial $(1-\varepsilon)$-EFX allocation carefully, we can also guarantee high Nash welfare for our final $(1-\varepsilon)$-EFX allocation with sublinear many goods unallocated. To this end, we use an important result from Caragiannis et al.~\cite{CaragiannisGravin19} about determining partial EFX allocations with high Nash welfare in polynomial time.
\begin{theorem}[\cite{CaragiannisGravin19}]\label{EFXwithNWCaragiannis}
	In polynomial time, we can determine a partial EFX allocation $X$ such that $\mathit{NW}(X) \geq 1/(2.88) \cdot \mathit{NW}(X^*)$ where $X^*$ is the Nash welfare maximizing allocation.\footnote{In fact, the result in~\cite{CaragiannisGravin19} show the existence of partial EFX allocations that achieve a $1/2$ approximation of the Nash welfare. However, in polynomial time, one can only find a partial EFX allocation with a $1/2.88$ approximation of the Nash welfare.}
\end{theorem}

Let $X$ be the partial EFX allocation that achieves a $2.88$ approximation of the Nash welfare. We run our algorithm starting with $X$ as the initial allocation. The final $(1-\varepsilon)$-EFX allocation with sublinear many unallocated goods is also a $2.88$ approximation of the Nash welfare as the valuations of the agents in the final allocation is at least their valuations in $X$. Therefore, we have the following theorem,

\begin{theorem}
	\label{EFXwithNashwelfare}
	In polynomial time, we can determine a $(1-\varepsilon)$-EFX allocation with $\mathcal{O}((n / \varepsilon)^{\tfrac{4}{5}})$ goods unallocated. Furthermore, $\mathit{NW}(X) \geq 1/(2.88) \cdot \mathit{NW}(X^*)$.\footnote{Note that using the existence of partial EFX allocations with $1/2$ approximation to Nash welfare, one can also claim the existence of a $(1-\varepsilon)$-EFX allocation $X$ with $\mathcal{O}((n / \varepsilon)^{\tfrac{4}{5}})$ goods unallocated such that $\mathit{NW}(X) \geq 1/2 \cdot \mathit{NW}(X^*)$.}  
\end{theorem} 

%\paragraph{$\ram(2) \leq 2$.}

\section{Limitations of the Approach in~\cite{CGM20}}
\label{limitationoftechnique}
In~\cite{CGM20}, an algorithmic proof to the existence of EFX allocations is shown for three agents with additive valuations. We briefly sketch the proof technique in~\cite{CGM20} and then highlight why it does not work for determining a $(1-\varepsilon)$-EFX allocations with just four agents. Let the three agents be $a$, $b$ and $c$ and for any allocation $X$, let $\phi(X)$ be the vector $\langle v_a(X_a), v_b(X_b), v_c(X_c) \rangle$. The algorithm starts with an empty allocation which is trivially  EFX and as long as there is an unallocated good, the algorithm determines another EFX allocation $X'$ such that $\phi(X')$ is lexicographically larger than $\phi(X)$, i.e., either $v_a(X'_a) > v_a(X_a)$ or $v_a(X'_a) = v_a(X_a)$ and $v_b(X'_b) > v_b(X_b)$ or $v_a(X'_a) = v_a(X_a)$, $v_b(X'_b) = v_b(X_b)$ and $v_c(X'_c) > v_c(X_c)$. Berger et al.~\cite{BCFF'21} show that the same potential (namely $\phi( \cdot )$) can be used to show the existence of EFX allocations for four agents with at most one unallocated good. In this paper, we show that such a technique cannot be used to show the existence of $(1-\varepsilon)$-EFX allocations for four agents.

\begin{theorem}
	\label{counterexampleEFX3}
     There exists an instance $I$ with four agents, $\{a,b,c,d\}$ with additive valuations,  nine goods $\{ g_i \mid i \in [9] \}$ and a partial $(1-\varepsilon)$-EFX allocation $X$ on the goods $\cup_{i \in [8]} g_i$, such that in all complete $(1-\varepsilon)$-EFX allocation, the valuation of agent $a$ will be strictly less than her valuation in $X$, i.e., for all complete $(1-\varepsilon)$-EFX allocation $Y$, $\phi(Y)$ is lexicographically smaller than $\phi(X)$.  	
\end{theorem}
   
We remark that our instance builds on the instance in~\cite{CGM20}, that is used to show the existence of a partial EFX allocation which is not Pareto-dominated by any complete EFX allocation.  We now construct an instance $I$ with four agents, say $\{a,b,c,d\}$ with additive valuations and nine goods $\{ g_i \mid i \in [9] \}$. Let $\phi(X) = \langle v_a(X_a), v_b(X_b), v_c(X_c), v_d(X_d) \rangle$. We show a $(1-\varepsilon)$-EFX allocation $X$ of eight goods among four agents. Then we show in any complete $(1-\varepsilon)$-EFX allocation, the valuation of agent $a$ will be strictly less than (almost half of) her valuation in $X$. This shows that for any complete $(1-\varepsilon)$-EFX allocation $Y$, we have $\phi(X)$ is lexicographically larger than $\phi(Y)$. 

\begin{table}[t]
	\begin{center}
		\begin{tabular}{||c c c c c c c c c c||} 
			\hline
			& $g_1$ & $g_2$ & $g_3$ & $g_4$ & $g_5$ & $g_6$ & $g_7$ & $g_8$ & $g_9$ \\ [0.5ex] 
			\hline\hline
			$\mathbf{a}$ & $0$ & $0$ & $0$ & $0$ & $0$ & $0$ & $6$ & $4$ & $0$ \\  
			\hline 
			$\mathbf{b}$ & $16$ & $4$ & $24$ & $4$ & $0$ & $34$ & $31$ & $0$  & $2$ \\ 
			\hline
			$\mathbf{c}$ & $10$ & $0$ & $18$ & $8$ & $20$ & $0$  & $29$ & $0$ &$6$ \\
			\hline
			$\mathbf{d}$ & $0$ & $0$ & $0$ & $0$ & $18$ & $20$ & $19$ & $0$ &$4$ \\
			\hline
		\end{tabular}
	\end{center}
	\caption{An instance where showing that the technique in~\cite{CGM20} cannot be used to determine $(1-\varepsilon)$-EFX allocations with four agents. In particular, given a $(1-\varepsilon)$-EFX allocation $X$ and the unallocated good $g_9$, there is no complete $(1-\varepsilon)$-EFX allocation where the valuation of agent $a$ does not strictly decrease, i.e., in any complete $(1-\varepsilon)$-EFX allocations $Y$, we have $v_a(Y_a) < v_a(X_a)$.}
	\label{Example One}
\end{table}

The full description of our instance is captured by Table~\ref{Example One}. We choose our $\varepsilon \ll 1$. The sub-instance defined by the agents $b$, $c$ and $d$, and the goods $\cup_{i \in [6]} g_i \cup g_9$ is the instance in~\cite{CGM20} used to show the existence of a partial EFX allocation which is not Pareto-dominated by any complete EFX allocation. We now specify the allocation $X$.
\begin{align*}
&X_a = \{g_7, g_8 \} &X_b = \{g_2, g_3, g_4 \}\\
&X_c = \{g_1, g_5 \} &X_d = \{g_6\}
\end{align*}
The good $g_9$ is unallocated. We will show that in any complete $(1-\varepsilon)$-EFX allocation, agent $a$ cannot have both $g_7$ and $g_8$.  This would imply that agent $a$'s valuation in any final $(1-\varepsilon)$-EFX allocation is strictly less than her valuation in $X$ (as agent $a$'s valuation for all goods other than $g_7$ and $g_8$ is zero). We prove this claim by contradiction. So assume that $Y$ is a complete $(1-\varepsilon)$-EFX allocation and $\{g_7, g_8\} \subseteq Y_a$. Note that $v_b(g_7) = 31$, $v_c(g_7) = 29$, and $v_d(g_7) = 19$. Since $Y_a$ contains at least one other good namely $g_8$, each of the agents $b$, $c$ and $d$ need to be allocated bundles that they value at least $31$, $29$ and $19$ respectively. 

First, consider the case that ${g_6 \in Y_b }$. Then we have $v_b(Y_b) \geq 34$. Now, to ensure $v_d(Y_d) \geq 19$, we need to allocate $g_5$ and $g_9$ to $d$, as $d$ values all the other goods zero. We are left with goods $g_1$, $g_2$, $g_3$ and $g_4$. In order to ensure $v_c(Y_c) \geq 29$, we definitely need to allocate $g_1$, $g_3$ and $g_4$ to $c$. Now, even if we allocate the remaining good $g_2$ to $b$, we have $v_b(Y_b) = v_b(\sset{g_2,g_6}) = 38 < (1-\varepsilon) \cdot 40 = (1-\varepsilon) \cdot v_b(\sset{g_1,g_3}) \leq (1-\varepsilon) \cdot {v_b}(Y_c \setminus g_4)$. Therefore, $b$ will strongly envy $c$. Thus $g_6 \notin Y_b$.

If $g_6 \notin Y_b$ and $v_b(Y_b) \geq 31$, {$Y_b$ must contain $g_3$} (the total valuation for $b$ of all the goods other than $g_3$, $g_6$, $g_7$ and $g_8$ is less than $31$). Now we consider some more subcases.

Let us first assume that ${g_1 \in Y_b}$. Since $Y_b$ already contains $g_1$ and $g_3$, the goods that can be allocated to $c$ and $d$ are $g_2$, $g_4$, $g_5$, $g_6$, and $g_9$. In order to ensure $v_c(Y_c) \geq 29$ we need to allocate $g_4$, $g_5$, and $g_9$ to $c$. Now, even if we allocate all the remaining goods ($g_2$ and $g_6$) to $d$, we have $v_d(Y_d) = v_d(\sset{g_3,g_6}) = 20 < (1-\varepsilon) \cdot 22 = (1-\varepsilon) \cdot v_d(\sset{g_5,g_7}) \leq (1-\varepsilon) \cdot v_d(Y_c \setminus g_4)$. Therefore, $d$ will strongly envy $c$. 

Thus $g_1 \notin Y_b$. Since neither $g_1$ nor $g_6$ belongs to $Y_b$, the only way to ensure $v_b(Y_b) \geq 31$ is to at least allocate $g_2$, $g_3$, and $g_4$ to $b$ (we can allocate more). Similarly, given that the goods not allocated yet are $g_1$, $g_5$, $g_6$, and $g_9$, the only way to ensure $v_c(Y_c) \geq 29$ is to allocate at least $g_1$ and $g_5$ to  $c$. Similarly, the only way to ensure $v_d(Y_d) \geq 19$ now is to allocate at least $g_6$ to $d$. Now we only have to allocate $g_9$. We  show that adding $g_9$ to any one of the existing bundles will cause a violation of the $(1-\varepsilon)$-EFX property.

\begin{itemize}
	\item Adding $g_9$ to $Y_a$: $b$, $c$ and $d$ strongly envies $a$ as $v_b(Y_b) = 32 < (1-\varepsilon) \cdot 33 = (1-\varepsilon) \cdot v_b(\{g_7,g_9\}) \leq (1-\varepsilon) \cdot v_b(Y_a \setminus g_8)$. Similarly we have $v_c(Y_c) = 30 < (1-\varepsilon) \cdot 35 = (1-\varepsilon) \cdot v_c(\{g_7,g_9\}) \leq (1-\varepsilon) \cdot v_c(Y_a \setminus g_8)$ and $v_d(Y_d) = 20 < (1-\varepsilon) \cdot 23 = (1-\varepsilon) \cdot v_d(\{g_7,g_9\}) \leq (1-\varepsilon) \cdot v_d(Y_a \setminus g_8)$. 
	\item Adding $g_9$ to $Y_b$: $c$ strongly envies $b$ as $v_c(Y_c)=30 < (1-\varepsilon) \cdot 32 = (1-\varepsilon) \cdot v_c(\left\{g_3,g_4,g_7\right\}) = (1-\varepsilon) \cdot v_c(Y_b \setminus g_2)$.
	\item Adding $g_9$ to $Y_c$: $d$ strongly envies $c$ as $v_d(Y_d)=20 < (1-\varepsilon) \cdot 22 = (1-\varepsilon) \cdot v_d(\left\{g_5,g_9\right\}) = (1-\varepsilon) \cdot v_d(Y_c \setminus g_1)$.
	\item Adding $g_9$ to $Y_d$: $b$ strongly envies $d$ as $v_b(Y_a)=32 < (1-\varepsilon) \cdot 34 = (1-\varepsilon) \cdot v_b(g_6) = (1-\varepsilon) \cdot v_b(Y_d \setminus g_9)$.
\end{itemize}

This shows that $\{g_7, g_8\} \not \subseteq Y_a$ for any complete $(1-\varepsilon)$-EFX allocation $Y$. This implies that agent $a$'s valuation in $Y$ is strictly less than her valuation in $X$, implying that $\phi(X)$ is lexicographically larger than  $\phi(Y)$. This shows that the approach from~\cite{CGM20} cannot be generalized to guarantee $(1-\varepsilon)$-EFX allocation when there are four or more agents. 

\clearpage

\bibliographystyle{abbrv}
\bibliography{EFX}

\end{document}